\documentclass{thesis}

\pdfoutput=1

\title{Verifying Strong Eventual Consistency in $\delta$-\CRDTs}
\author{Taylor Blau}
\supervisor{Dan Grossman}

\begin{document}
  \pagenumbering{gobble}

  \newgeometry{hmargin=1in,centering}
    \begin{titlepage}
  \centering
  \makeatletter
    {\Huge \@title \par}
  \makeatother
  \vskip 4em
  by
  \vskip 3em
  {
    \lineskip .5em
    \begin{tabular}[t]{c}
      { \Large Taylor Blau } \\ \\
      { Supervised by Dan Grossman }
    \end{tabular}\par}
  \vskip 5em
  {
    A senior thesis submitted in partial fulfillment of \\
    \vskip 1em
    the requirements for the degree of
  }
  \vskip 3em
  {
    \large
    Bachelor of Science \\
    With Departmental Honors \\

    \vskip 1em

    Computer Science \& Engineering \\

    \vskip 1em

    University of Washington \\

    \vskip 1em

    June 2020
  }
  \vfill
  {
    \flushleft
    \noindent
    Presentation of work given on~\makebox[3in]{\hrulefill} \\
    \vskip 1em

    \noindent
    Thesis and presentation approved by~\makebox[3in]{\hrulefill} \\
    \vskip 1em

    \noindent
    Date~\makebox[2in]{\hrulefill} \\
  }
\end{titlepage}

  \restoregeometry

  \frontmatter
  \cleardoublepage
  \pagenumbering{roman}

  \chapter*{Abstract}

Conflict-free replicated data types (\CRDTs) are a natural structure with which
to communicate information about a shared computation in a distributed setting
where coordination overhead may not be tolerated, and individual participants
are allowed to temporarily diverge from the overall computation.  Within this
setting, there are two classical approaches: state- and operation-based \CRDTs.
The former define a commutative, associative, and idempotent \textit{join}
operation, and their states a \textit{monotone join semi-lattice}. State-based
\CRDTs may be further distinguished into classical- and $\delta$-state \CRDTs.
The former communicate their \emph{full} state after each update, whereas the
latter communicate only the \emph{changed} state. Op-based \CRDTs communicate
\emph{operations} (not state), thus making their updates non-idempotent.
Whereas op-based \CRDTs require little information to be exchanged, they demand
relatively strong network guarantees (exactly-once message delivery), and
state-based \CRDTs suffer the opposite problem. Both satisfy \textit{strong
eventual consistency} (\SEC).

We posit that $\delta$-state \CRDTs both (1) require less communication overhead
from payload size, and (2) tolerate relatively weak network environments, making
them an ideal candidate for real-world use of \CRDTs. Our central intuition is a
pair of reductions between state-, $\delta$-state, and op-based \CRDTs. We
formalize this intuition in the Isabelle interactive theorem prover and show
that state-based \CRDTs achieve \SEC. We present a relaxed network model in
Isabelle and show that state-based \CRDTs still maintain \SEC. Finally, we
extend our work to show that $\delta$-state \CRDTs maintain \SEC when only
communicating $\delta$-state fragments, even under relatively weak network
conditions.

  \chapter*{Acknowledgements}

This thesis is the product of many ideas grown out of collaboration and
discussion with my advisory committee, as well as other researchers in this
area.

First, Talia Ringer, my senior thesis mentor. Talia's thoughtfulness and
willingness to absorb a new research area was inspiring and fostered me to
look at this area from a new angle. Her patience in acquainting me with
interactive theorem provers was key in making this thesis possible. Though
always a source of good ideas, this thesis would not exist without Talia's
unwavering support. I would be remiss if I did not mention Talia's encouragement
throughout, even when the process was overwhelming.

Second, Dan Grossman, my faculty advisor. Dan has made my undergraduate
experience meaningful in ways that I am not sure many others are as fortunate as
I to have experienced. Dan took a skeptical pre-freshman, encouraged him to take
CSE 341, and indulged him in many walks back to the Paul G. Allen building after
class. Dan allowed me to T.A. for him, and was unflapped when I informed him
that I had volunteered him to be my faculty advisor.\footnote{It could be said
he was volun\textit{told}.} Of course, Dan is also a font of insight, offering
new ideas and perspectives when they were needed, and always giving me something
to think about after our meetings.

I would also like to thank Martin Kleppman, as well as his co-authors, for his
constant correspondence throughout this work. Their work is foundational to our
approach, and is the basis on which many of our ideas (and proofs) are built.
Martin was always willing to discuss the state of our work, and to offer his
guidance about interesting directions to pursue.

Finally, I wish to thank my family. My Mom and Dad, for their love, for always
encouraging me, and for giving me the freedom to explore areas that interested
me. Tracy and Richard Lippard, for their encouragement and hospitality during
which significant portions of this thesis were written. Lastly, I wish to thank
Maya Lippard, my partner, constant source of inspiration, and without whom this
thesis would not exist.

  \chapter*{Dedication}

\begin{flushright}
  \textit{To Maya, forever and ever.}
\end{flushright}

  \tableofcontents
  \listoffigures

  \mainmatter
  \chapter{Introduction}

Computational systems today are larger than ever. Whereas previously one would
architect their programs to run on a single system, it is now commonplace to
design programs that share computation across multiple machines which
communicate with each other in a coordinated fashion. Therefore, it is natural
to ask why one might design from the latter perspective rather than the former.
The answer is threefold:
\begin{enumerate}
  \item \emph{Resiliency}. Designing a computational workload to be distributed
    among participants tolerates the failure of any one (or more) of those
    participants.
  \item \emph{Scalability}. When designed from a distributed standpoint,
    ``scaling'' your workload to meet a higher demand is reduced to adding
    additional hardware, not designing more efficient ways to do the
    computation.
  \item \emph{Locality}. When a system is accessed from a broad set of
    geographic locations, strategic placement of hardware in locations near
    request-origin sites can lower latency for users.
\end{enumerate}

So, it is clear that as our demand on such computations grow, that so too
must our need to design these systems in a way that first considers the concerns
of resiliency, scalability, and locality.

In order to design systems in this way, however, one must consider additionally
the challenges imposed by not having access to shared memory among participants
in the computation. If a program runs in a single-threaded fashion on a single
computer, there is no need to coordinate memory accesses, since only one part of
the program may read or write memory at a given time. If the program is written
to be multithreaded, then the threads must coordinate among themselves by using
mutexes or communication channels to avoid race conditions and other concurrency
errors.

The same challenge exists when a system is distributed at the hardware and
machine level, rather than among multiple threads running on a single piece of
hardware. The challenge, however, is made more difficult by the fact that the
communication overhead is far higher between separate pieces of hardware than
between two threads.

This thesis focuses on datatypes by which computation can be coordinated across
multiple machines. In particular, we formalize a set of consistency guarantees
(namely, Strong Eventual Consistency, hereafter \SEC) over a class of replicated
datatypes, $\delta$-state Conflict-Free Replicated Datatypes (\CRDTs). We
describe the preliminaries necessary to contextualize the body of this work in
the following section.

\section{Preliminaries}
Our discussion here focuses on \CRDTs, which are designed to be
both easily distributed and require relatively low coordination overhead by
allowing individual participants to diverge temporarily from the state of the
overall computation. That is, the computation reflects a different value
depending on which participant in the computation responds to the request.

These datatypes operate in such a way so as to both avoid conflict between
concurrent updates, and to avoid locking and coordination
overhead~\citep{shapiro11}. \CRDTs have seen moderate use in industry.
Based on introspection of the runtime headers in iOS, Apple is believed to use
\CRDTs for offline synchronization of content in their note-taking
app, Notes~\citep{applenotes}.  Redis, a popular open-source distributed cache
uses \CRDTs in their Enterprise offering to perform certain kinds of replication
and conflict-resolution~\citep{redis}.

\CRDTs are said to achieve \SEC which is to say that they achieve a stronger
form of \textit{eventual consistency} (\EC). We summarize the definitions of
eventual- and strong eventual consistency from~\cite{shapiro11}.

\begin{definition}[Eventual Consistency]
  \label{def:eventual-consistency}
  A replicated datatype is \emph{eventually consistent} if:
  \begin{itemize}
    \item Updates delivered to it are eventually delivered to all other replicas
      in the system.
    \item All well-behaved replicas that have received the same set of updates
      eventually reflect the same state.
    \item All executions on this datatype are terminating.
  \end{itemize}
\end{definition}

\begin{definition}[Strong Eventual Consistency]
  A replicated datatype is \emph{strong eventually consistent} if:
  \begin{itemize}
    \item It is eventually consistent, as above.
    \item Convergence occurs immediately, that is, any two replicas that have
      received the same set of updates \emph{always} reflect the same state.
  \end{itemize}
\end{definition}

Broadly speaking, there are two classes of \CRDTs, which we refer to as the op-
and state-based variants. We will provide formal definitions for each of the two
classes in Chapter~\ref{chap:background}. We now present brief definitions of
op- and state-based \CRDTs based on~\citet{baquero14} and~\citet{shapiro11}:

\begin{definition}[Operation-based Conflict-Free Replicated Datatype (op-based
\CRDT)]
  op-based \CRDTs apply updates in two phases:
  \begin{enumerate}
    \item First, an operation is \emph{prepared} locally. At this phase, the
      op-based \CRDT combines the operation with the current state to send a
      representation of the update to other replicas.
    \item Then, the represented operation is applied to other replicas using
      \emph{effect}, where \emph{effect} is commutative for concurrent
      operations.
  \end{enumerate}
\end{definition}

\begin{definition}[State-based Conflict-Free Replicated Datatype (state-based
\CRDT)]
  state-based \CRDTs only apply updates to their local state, and periodically
  send serialized representations of the contents of their state to other
  replicas.

  Crucially, these states form a \textit{monotone join semi-lattice} (i.e,. a
  lattice $\langle S, \sqcup \rangle$ where for any $s_1, s_2 \in S$ at both
  $s_1 \sqsubseteq s_1 \sqcup s_2$ and $s_2 \sqsubseteq s_1 \sqcup s_2$ hold for
  commutative, associative, and idempotent $\sqcup$).

  To achieve convergence, state-based \CRDTs periodically send their state to
  other replicas, which then replace their own state by joining the received
  state into their own.
\end{definition}

\section{op- and state-based trade-offs}

These two classes are distinguished from one another based on their strengths
and weaknesses. In one sense, op- and state-based \CRDTs form a kind of a dual,
where they trade off strong network guarantees for message payload
size~\citep{baquero14}.

Because the state-based \CRDT needs to send a representation of its entire
state, it often requires a significant amount of network bandwidth to propagate
large messages~\citep{almedia18}. In Section~\ref{sec:example-gcounter} we will
present an example where the payload size grows as a linear function of the
number of replicas. In return for this large payload size, state-based \CRDTs
are able to achieve \SEC even in networks that are allowed to drop, reorder, and
duplicate messages.

On the other hand, op-based \CRDTs require relatively little network bandwidth
to send a notification of a single update (typically the representation
generated in the \textit{prepare} stage is dwarfed by the typical payload size
of a state-based \CRDT), but in exchange demand that the network deliver
messages in-order for sequential (comparable) updates and
at-most-once delivery~\citep{shapiro11}.

Significant work in this area (\cite{almedia18, enes18, cabrita17,
vanDerLinde16}) has focused on mediating these two extremes. This line of
research (particularly in~\citet{almedia18}) has identified $\delta$-state
\CRDTs---a variant of the state-based \CRDT which we discuss in
Section~\ref{sec:state-based-crdts}---as an alternative which occupies a
satisfying position between the two extremes.  $\delta$-state \CRDTs behave as
traditional state-based \CRDTs, with the exception that their updates consist of
state \emph{fragments} instead of their entire state. These fragments (generated
by $\delta$-mutators and called $\delta$-updates) are then applied locally at
all other replicas to reassemble the full state. Because these fragments often
do not need to comprise the full state, $\delta$-state \CRDTs in general have
small payload size (thus requiring a similar amount of bandwidth as messages
sent and received from op-based \CRDTs), while still tolerating the same set of
network deficiencies as state-based \CRDTs. This combination of properties makes
them an appealing alternative to traditional state- and op-based \CRDTs, and
places interest in studying their convergence properties.

\section{Contributions}

Our main contribution builds on the work in~\citet{gomes17} and introduces a set
of formally verified, machine-checked proofs in Isabelle~\citep{wenzel02} of the
main result in~\citet{almedia18}, which we re-state below:\footnote{The source
of our proofs is available for free at:
\url{https://github.com/ttaylorr/thesis}.}

\begin{theorem}[Almedia, Shoker, Baquero, '18]
  Consider a set of replicas of a $\delta$-\CRDT object, replica $i$ evolving
  along a sequence of states $X_i^0 = \bot$, $X_i^1=\ldots$, , each replica
  performing delta-mutations of the form $m^\delta_{i,k}(X^k_i)$ at some subset
  of its sequence of states, and evolving by joining the current state either
  with self-generated deltas or with delta-groups received from others. If each
  delta-mutation $m^\delta _{i,k}(X^k_i)$ produced at each replica is joined
  (directly or as part of a delta-group) at least once with every other replica,
  all replica states become equal.
\end{theorem}

Here, $X_i^t$ refers to the state of the $i$th replica at time $t$, and
$m^\delta_{i,k}(X_i^k)$ refers to the $\delta$-mutation applied at the $i$th
replica at time $k$.

We rely on the work of~\citet{gomes17} in order to build a handful of state- and
$\delta$-state \CRDTs as in~\citet{almedia18} to show that even under weak
network guarantees\footnote{We inherit dropping and reordering of messages from
the original work of~\citet{gomes17}, but further relax the network model by
also allowing messages to be duplicated.} these $\delta$-state \CRDTs still
achieve \SEC.

Our verification efforts yielded a pair of \CRDTs---the grow-only counter
(G-Counter) and set (G-Set)---in three encodings: one state-based, and two
$\delta$-state encodings. Our key idea guiding these verification efforts is to
treat op- and state-based \CRDTs similarly by modeling state-based \CRDTs as
op-based where the operation is the join provided by the
semi-lattice.\footnote{This approach is described in detail in
Section~\ref{sec:state-as-op}.} We show that \SEC is preserved in these \CRDTs,
even when the underlying $\isa{network}$ interface has been weakened
substantially from when it was introduced in the aforementioned work.

The remainder of this thesis is ordered as follows:
\begin{itemize}
  \item In Chapter~\ref{chap:background}, we summarize existing research in the
    broader realm of \CRDTs. We present formal definitions of op- and
    state-based \CRDTs, and conduct a thorough discussion of their relative
    strengths and weaknesses. Likewise, we present a summary of some work in the
    area of $\delta$-state \CRDTs, and present its strengths.
  \item In Chapter~\ref{chap:crdt-instantiations}, we discuss examples of two
    \CRDTs in an op-, state-, and $\delta$-state style. These objects will be
    the subject of our verification efforts in Chapter~\ref{chap:example-crdts}.
  \item In Chapter~\ref{chap:crdt-reductions}, we outline a pair of reductions
    between state-, op-, and $\delta$-state based \CRDTs which guides the
    majority of our proof strategy.
  \item In Chapter~\ref{chap:example-crdts}, we discuss the outcome of our
    approach by presenting a pair of successfully-verified $\delta$-state
    \CRDTs, as well as describe our efforts in relaxing the network model in
    order to verify these objects over a non-trivial set of network behaviors.
  \item In Chapter~\ref{chap:future-work}, we suggest future
    research directions. We consider a handful of areas in which formalizing
    existing results may be fruitful, as well as a handful of additional
    approaches to the proofs we presented here.
  \item In Chapter~\ref{chap:conclusion}, we conclude.
\end{itemize}

  \chapter{Background}
\label{chap:background}

This chapter outlines the preliminary information necessary to contextualize the
remainder of this thesis for readers unfamiliar with existing \CRDT research.
Here we motivate \CRDTs, formalize their state- and op-based variants, and
present examples of common instantiations. Finally, we conclude with a
discussion of the different levels of consistency guarantees that each \CRDT
variant offers, and rationalize which levels of consistency are appealing in
certain situations.

\section{Motivation}
\CRDTs are a way to store several copies of a data-structure on multiple
computers which form a distributed system. Each participant in the system can
make modifications to the datatype without the need for explicit coordination
with other participants. \CRDT implementations are designed so that
coordination-free updates which may conflict with one another always have a
deterministic resolution. This allows multiple participants to query and modify
their \emph{view} of the replicated datatype, without the traditional overhead
and implementation burden that more stringent replication algorithms require.

Here, we'll discuss three variants of \CRDTs: state-based, op-based, and
$\delta$-state based. Each of these variants achieve a consistent value by the
use of different message types, and each likewise requires a different set of
delivery semantics. In this chapter, we identify $\delta$-state \CRDTs as
achieving an appealing set of trade-offs among each of the three variants. We
restate that they are able to achieve \SEC (the best reasonably-achievable
consistency guarantee for most \CRDT applications) while maintaining both:
\begin{itemize}
  \item A relatively small payload size, as is the benefit of op-based \CRDTs,
    and
  \item Relatively weak delivery semantics, as is the benefit of state-based
    \CRDTs.
\end{itemize}

\section{Coordinated Replication}
In a distributed system, it is common for more than one participant to need to
have a \textit{view} of the same data. For example, multiple nodes may need to
have access to the same internal data structures necessary to execute some
computation. When a piece of data is shared among many participants in a system,
we say that that data is \textit{replicated}.

However, saying only that some data is ``replicated'' is underspecified. For
example: how often is that data updated among multiple participants? How does
that data behave when multiple participants are modifying it concurrently? Do
all participants always have the same view of the data, or are there temporary
divergences among the participants in the system?

It turns out that the answer to the last question is of paramount importance.
Traditionally speaking, in a distributed system, all participants have an
identical replica of any piece of shared data at all times. That is, at no
moment in time will there be a replica that could atomically compare its
replicated value for some data with any other replica for equality and disagree.
Said otherwise, all replicated values are equal everywhere all at once. This is
an appealing property to say the least, because it allows system designers to
conceptually treat a distributed system as a single unit of computation. That
is, if all replicas maintain the same memory, it is conceptually as if one whole
machine is being replicated many times.

That being said, upholding this requirement is not a straightforward task. Some
question that arise are: who coordinates when updates to a piece of data are
replicated to other participants in the system? What happens when the
coordinator becomes unresponsive, or otherwise misbehaves? Who is responsible
for electing a new participant to take over the coordination duties of the
participant which was no longer able to fulfill them?

\section{Distributed Consensus Algorithms}
\label{sec:dca-safety}

These questions give rise to the area of consensus algorithms. Broadly speaking,
a consensus algorithm is a routine which multiple participants follow in order
to agree on a shared value.

We first state briefly the properties that an algorithm must have to solve
distributed consensus from~\citet{howard19}:
\begin{definition}[Distributed Consensus Algorithm]
  \label{def:consensus}
  An algorithm is said to solve distributed consensus if it has the following
  three safety requirements:
  \begin{enumerate}
    \item \emph{Non-triviality}: The decided value must have been proposed by a
      participant.
    \item \emph{Safety}: Once a value has been decided, no other value will be
      decided.
    \item \emph{Safe learning}: If a participant learns a value, it must learn
      the decided value.
  \end{enumerate}
  In addition, it must satisfy the following two progress requirements:
  \begin{enumerate}
    \item \emph{Progress}: Under previously agreed-upon liveness conditions, if
      a value is proposed by a participant, then a value is eventually decided.
    \item \emph{Eventual learning}: Under the same conditions as above, if a
      value is decided, then that value must be eventually learned.
  \end{enumerate}
\end{definition}

The two most popular algorithms in this field are Paxos and
Raft~\citep{howard20,lamport98,ongaro14}. Each implements distributed
state-machine replication and can be used to implement linearizable systems.
Both of these systems are notoriously difficult to understand and implement
correctly in practice~\citep{howard20}. The topics often appear in
undergraduate-level courses in Distributed Systems, and have been the subject of
extensive verification effort to date~\citep{wilcox15}. Often, these distributed
systems verification efforts require an enormous amount of effort.  In a
companion paper~\citet{woos16} use on the order of 45,000 lines of proof scripts
to verify the complete Raft protocol in their system.

It is natural to ask what is the property of these systems that makes them
difficult to implement or reason about correctly in practice. One possible
answer is to look at the stringent safety requirements (that is, that once a
value has been decided, no other value(s) will be decided) in these algorithms.

\CRDTs are a natural response to this. By allowing participants to temporarily
diverge from the state of the overall computation (cf., the second property of
Definition~\ref{def:eventual-consistency}), \CRDTs allow replicas to violate the
safety property of Definition~\ref{def:consensus}.  By giving up the immediacy
and permanence that the safety properties of a traditional distributed consensus
algorithm, \CRDTs allow for a dramatically lower implementation burden in
practice, and are substantially easier to reason about.

\section{Consistency Guarantees}
\CRDTs are said to attain a weaker form of consistency known as \emph{strong
eventual consistency}~\citep{shapiro11}. \SEC is a refinement of \emph{eventual
consistency} (\EC). Informally, \EC says that reads from a system eventually
return the same value at all replicas, while \SEC says that if any two nodes
have received the same set of updates, they will be in the same state.

\EC and the \SEC extension are natural answers to the question we pose in
Section~\ref{sec:dca-safety}. That is, we posit that it is the safety
requirement in traditional Distributed Consensus Algorithms which make them
difficult to implement correctly. \EC makes only a liveness guarantee, and so on
its own it is not a sufficient solution for handling distributed consensus in an
environment with relaxed requirements. \SEC, however, does add a safety
guarantee, but the precondition (namely that only nodes which have received the
same \emph{set} of updates will be in the same state) makes it possible to relax
our requirements around network delays, or particulars of a \CRDT algorithm
which do not send updates to all other replicas immediately.

In short, we believe that it is this relaxation--that is, that \CRDTs are only
required to be in the same state \emph{eventually}, conditioned on which updates
they have and have not yet received--which makes \SEC an appealing consistency
property for distributed systems which more relaxed requirements than would be
satisfied by a linearizable system.

We discuss each of these consistency classes in turn.

\subsection{Eventual Consistency}
\EC captures the informal guarantee that if all clients stop submitting updates
to the system, all replicas in the system eventually reach the same
value~\citep{shapiro11}. More formally, \EC requires the following three
properties~\citep{shapiro11}:
\begin{enumerate}
  \item \emph{Eventual delivery}. An update delivered at some correct replica is
    eventually delivered at all replicas.
    \[
      \forall r_1, r_2.\, f \in (\textsf{delivered}~r_1) \Rightarrow \Diamond f
      \in (\textsf{delivered}~r_2)
    \]
  \item \emph{Convergence}. Correct replicas which have received the same
    \emph{set} of updates eventually reflect the same state.
    \[
      \forall r_1, r_2.\,~\square~(\textsf{delivered}~r_1) =
      (\textsf{delivered}~r_2) \Rightarrow \Diamond~\square~q(r_1) = q(r_2)
    \]
  \item \emph{Termination}. All method executions terminate.
\end{enumerate}

(For readers unfamiliar with modal logic notation, we use $\Diamond$ to precede
a logical statement that is true at \emph{some} time, whereas we use $\square$
to precede a logical statement that is true at \emph{all} times.)

\EC is a relatively weak form of consistency. In~\citet{shapiro11}, it is
observed that \EC systems will sometimes execute an update immediately only to
discover that it produces a conflict with some future update, and so frequent
roll-backs may be performed. This imposes an additional constraint, which is
that replicas need to form consensus on the ``standard'' way to resolve
conflicts so that the same conflicts are resolved identically at different
replicas.

We devote some additional discussion to the first property of \EC. Eventual
delivery requires that all updates delivered to some correct replica are
eventually delivered to all other correct replicas. This property alone permits
too much of the underlying network, and so it can make it difficult to reason
about strong consistency guarantees over an unreliable network.

Take for an example a network which never delivers any messages. In this case,
the precondition for eventual delivery is not met, and so we are relieved of the
obligation to prove that updates are propagated to other replicas, since they
aren't delivered anywhere in the first place. However, consider a network which
delivers only the \emph{first} message sent on it, and then drops all other
messages. In this case, it \emph{is} possible that a replica will receive some
update, attempt to propagate it to other replicas, only for them to never be
delivered.

To resolve this conflict in practice, one of two approaches is often taken. In
the first approach, assume a fair-loss network~\citep{cachin11} in which each
message has a non-zero probability of being delivered. To ensure that messages
are delivered, each node sends each message an infinite number of times over the
network, such that it will be delivered an infinite number of
times.\footnote{This approach is due to Martin Kleppman over e-mail, but can
also be found in the literature, for eg.,~\citet{shapiro11}.} This
resolves the eventual delivery problem since we assumed a sufficient (but
weaker) condition of the underlying network, and then showed it is possible to
implement eventual delivery on top of these network semantics.

In the second approach, we first consider a set of delivery semantics $P$ which
predicates allowed and disallowed network behaviors. Typically, $P$ is assumed
to preserve causal order.\footnote{This is a standard
assumption~\citep{shapiro11,gomes17}, and can be implemented by assigning a
vector-clock and/or globally-unique identifier (UID) to each message at the
network layer.} We then refine $P$ to ensure that the properties of \EC (and
\SEC) can be implemented on top of the network, resolving our problem by
discarding degenerate network behaviors.

\subsection{Strong Eventual Consistency}
Another downside of implementing a system which only upholds \EC is that \EC is
merely a liveness guarantee. In particular, \EC does not impose any restriction
on nodes which have received the same set or even sequence of messages. That is,
a pair of replicas which have received the exact set of messages in the exact
same order are not required to return the same value.

\SEC addresses this gap by imposing a safety guarantee in addition to the
previous liveness guarantees in \EC. That is, a system is \SEC when the
following two conditions are met:
\begin{enumerate}
  \item The system is \EC, per above guidelines.
  \item \emph{Strong convergence}. Any pair of replicas which have received the
    same set of messages must return the same value when queried immediately.
    \[
      \forall r_1, r_2.\, (\textsf{delivered}~r_1) = (\textsf{delivered}~r2)
        \Rightarrow q(r_1) = q(r_2)
    \]
\end{enumerate}

That is, it is the strong convergence property of \SEC that distinguishes it
from \EC. On top of \EC, strong convergence is only a moderate safety
restriction. In particular, it imposes no requirements on replicas which have
not received the same sequence or even set of updates. So, unlike strong
distributed consensus algorithms like Paxos or Raft which are fully
linearizable~\citep{lamport98,ongaro14}, \SEC allows certain replicas to be
``behind.'' That is, a replica which hasn't yet received all relevant updates in
the system is allowed to return an earlier version of the computation.

Informally, this means that replicas in the system are allowed to temporarily
diverge from the state of the overall computation. As soon as no more updates
are sent to the system, property (1) of \EC requires that all replicas will
\emph{eventually} converge to a uniform view of the computation.

\section{state-based \CRDTs}
\label{sec:state-based-crdts}

Now that we have discussed \EC and \SEC, we will turn our attention to datatypes
that implement these consistency models. \CRDTs are a common way to implement
the consistency requirements in \SEC. So, we begin with a discussion of
state-based \CRDTs from their inception in~\citet{shapiro11}. A state-based
\CRDT is a 5-tuple $(S, s^0, q, u, m)$. An individual replica of a state-based
\CRDT is at some state $s^i \in S$ for $i \ge 0$, and is initially $s^0$. The
value may be queried by any client or other replica by invoking $q$. It may be
updated with $u$, which has a unique type per \CRDT object. Finally, $m$ merges
the state of some other remote replica.  Neither $q$ nor $u$ have pre-determined
types, per se, rather they are implementation specific. We discuss a pair of
examples to illustrate this point in Chapter~\ref{chap:crdt-instantiations}.

Crucially, the states of a given state-based \CRDT form a partially-ordered set
$\langle S, \sqsubseteq \rangle$. This poset is used to form a join
semi-lattice, where any finite subset of elements has a natural least
upper-bound. Consider two elements $s^m, s^n \in S$. The least upper-bound
$s = s^m \sqcup s^n$ is given as:
\[
  \forall s'.\; s' \sqsupseteq s^m, s^n \Rightarrow
    s^m \sqsubseteq s \land
    s^n \sqsubseteq s \land
    s \sqsubseteq s'
\]
In other words, a $s = s^m \sqcup s^n$ is a least upper-bound of $s^m$ and $s^n$
if it is the smallest element that is at least as large as both $s^m$ and $s^n$.

\subsection{Merging states}

For now, we set aside $q$ and $u$, and turn our attention towards the merging
function $m$. $m$ resolves the states of two \CRDTs into a new state, which is
then assigned at the replica performing the merge. Given a suitable set of
states which forms a lattice, we assume that:
\[
  m(s_1, s_2) = s_1 \sqcup s_2
\]
for some join semi-lattice with join operation $\sqcup$, and that whenever a
\CRDT replica $r_1$ at state $s_1$ receives an update from another replica
$r_2$ at state $s_2$, that $r_1$ attains a new state $s_1' = m(s_1, s_2)$.
This process, in addition to each replica periodically broadcasting an update
which contains its current state, is carried on continually, and $m$ is
invoked whenever a new state is received. That is, each replica is evolving
over time in response to outside instruction, and in turn these updates cause
internal state transitions, which themselves cause those new states to be
broadcast and eventually joined at every other replica.

The $\sqcup$ operator has three mathematical properties that make it an
appealing choice for joining states together as in $m$. These are its
\emph{commutativity}, \emph{associativity}, and \emph{idempotency}. That is, for
any states $s_1$, $s_2$, and $s_3$, that:
\begin{itemize}
  \item The operator is \emph{commutative}, i.e., that $s_1 \sqcup s_2 = s_2
    \sqcup s_1$, or that order does not matter.
  \item The operator is \emph{idempotent}, i.e., that $(s_1 \sqcup s_2)
    \sqcup s_2 = s_1 \sqcup s_2$, or that repeated updates reach a fixed point.
  \item Finally, the operator is \emph{associative}, i.e., that $s_1 \sqcup (s_2
    \sqcup s_3) = (s_1 \sqcup s_2) \sqcup s_3$, or that grouping of arguments
    does not matter.
\end{itemize}

These mathematical properties correspond to real-world constraints that often
arise naturally in the area of distributed systems. We provide examples for each
of these three properties below:

\paragraph{Commutativity} Take, for example, that messages may occur out of
order. This often happens in, for example, UDP (User Datagram Protocol)
networks, where the received datagrams are not guaranteed to be in the order
that they were sent. Because $\sqcup$ is commutative, replicas joining the
updates of other replicas do not need to receive those updates in order, because
the result of $s_1 \sqcup s_2$ is the same as $s_2 \sqcup s_1$.  That is, it
does not matter which of two updates from another replica arrives first, because
the result is the same no matter in which order they are delivered.

For concreteness, say that we have two replicas, $r_1$ and $r_2$. $r_1$
initially begins at state $s$, and $r_2$ progresses through states $s_1, \ldots,
s_n$ for $n > 0$. We then see that it does not matter the order in which these
updates are delivered to $r_1$. Suppose that we have a bijection $\pi : [n] \to
[n]$ which maps the true order of a state $s_i$ to the order in which it was
delivered. Then, we can see that the choice of $\pi$ is arbitrary, because:
\[
  s \gets s \sqcup (s_{\pi(1)} \sqcup \cdots \sqcup s_{\pi(n)})
\]
for any choice of $\pi$, because
\[
  s_{\pi(1)} \sqcup \cdots \sqcup s_{\pi(n)} = s_1 \sqcup \cdots \sqcup s_n
\]
which follows from the fact that $\sqcup$ is commutative. This can be shown
inductively on the number of updates, $n$, given the commutativity of $\sqcup$.

\paragraph{Idempotency} Next, it is often common for packets to be duplicated in
transit over a network.  That is, even though a packet may be sent from a source
only once, it may be received by a recipient on the same network multiple times.
For this, the idempotency of $\sqcup$ comes in handy: no matter how many times a
state is broadcast from an evolving replica, any other replica on the network
will tolerate that set of messages, because it only requires the message to be
delivered once. Any additional duplicates are merged in without changing the
state.

\paragraph{Associativity} Finally, associativity is an appealing property, too,
although its applications are both less immediate and less often-used in this
thesis. Suppose that several replicas of a state-based \CRDT reside on a network
with, say, high latency, or it is otherwise undesirable to send more messages on
the network than is necessary. Because associativity implies that the grouping
of updates is arbitrary, a replica can maintain a \textit{set} of pending
updates, and periodically send that set to other replicas by first folding
$\sqcup$ over it and sending a single update.\footnote{``Periodically'' is
arbitrary and is left up to the implementation, but it would be easy to imagine
that this could be interpreted as whenever the set reaches a certain size,
and/or after a certain amount of time has passed since flushing the set of
pending updates.}

\section{op-based \CRDTs}
\label{sec:op-based-crdts}

Operation-based (op-based) \CRDTs evolve their internal states over time, but
these states need not necessarily form a semi-lattice.  Likewise, the
communication style of op- and state-based \CRDTs differ fundamentally: op-based
\CRDTs communicate \textit{operations} that indicate a kind of update to be
applied locally, instead of the \textit{result} of that update (as is the case
in state-based \CRDTs).

An op-based \CRDT is a $6$-tuple $(S, s^0, q, t, u, P)$. As in
Section~\ref{sec:state-based-crdts}, $S$, $s^0$, and $q$, retain their
original meaning (that is, the state set, an initial state, and a query
function).  In op-based \CRDTs, the pair $(t,u)$ takes the place of the $m$
merging function from state-based \CRDTs. $t$ and $u$ correspond to
\textit{prepare-update} and \textit{effect-update}, respectively. When an update
is made by a caller (say, for example, incrementing the value of an op-based
\CRDT counter), it is done in two phases~\citep{shapiro11}:
\begin{enumerate}
  \item First, the \textit{prepare-update} implementation $t$ is applied at the
    replica receiving the update. $t$ is side-effect free, and prepares a
    representation of the operation about to take place.
  \item Then, the \textit{effect-update} implementation $u$ is applied at the
    local and remote replicas if and only if the delivery precondition $P$ is
    met, causing the desired update to take effect. $P$ is interpreted
    temporally~\citep{shapiro11}, and is a precondition on whether or not
    operations necessary to process the \emph{current} operation have already
    been incorporated into the \CRDT's state. $P$ is traditionally assumed to be
    disabled until all messages which happened before the current message have
    been delivered, preserving causality.
\end{enumerate}

This is the critical distinction between op- and state-based CRDTS:
state-based \CRDTs propagate their state by applying a local update and taking
advantage of the lattice structure of their state-space in order to define a
convenient merge function. On the other hand, op-based \CRDTs propagate their
state by sending the \textit{representation} of an update to other replicas as
an instruction. This critical juncture translates into a corresponding
relaxation in the operation $(t, u)$, which is that unlike the state-based \CRDTs
whose $m$ must be commutative, associative, and idempotent, and op-based \CRDT
implementation of $(t, u)$ need only be commutative.

To explain why, we briefly restate the definition of a causal history for
op-based \CRDTs:

\begin{definition}[op-based Causal History~\citep{shapiro11}]
An object's casual history $C = \{ c_1, \ldots, c_n \}$ is defined as follows.
Initially, $c_i^0 = \emptyset$ for all $i \in \mathcal{I}$. If the $k$th method
execution is idempotent (that is, it is either $q$ or $t$), then the causal
history remains unchanged in the $k$th step, i.e., that $c_i^k = c_i^{k-1}$. If
the execution at $k$ is non-idempotent (i.e., it is $u$), then $c_i^{k} =
c_i^{k-1} \cup \{ u_i^k(\cdot) \}$.
\end{definition}

Causal history of an op-based \CRDT is defined based on the
\textit{happens-before} relation $\to$ as follows. An update $(t,u)$ happens
before $(t',u')$ (i.e., that $(t, u) \to (t', u')$) iff $u \in c_{j}^{K_j(t')}$
if $K_j$ is the injective mapping from operation to execution time. Shapiro and
his co-authors go on to describe a sufficient definition for the commutativity
of $(t,u)$ in op-based \CRDTs. In effect, they say that two pairs $(t,u)$ and
$(t',u')$ commute if and only if for any reachable state $s \in S$ the effect of
applying them in either order is the same. That is, $s \circ u \circ u' \equiv s
\circ u' \circ u$.

They claim that having commutativity for concurrent operations as well as an
in-order delivery relation $P$ for comparable updates is sufficient to prove
that op-based \CRDTs achieve \SEC.

\section{$\delta$-state \CRDTs}
In this section, we describe the refinement of \CRDTs that is the interest and
focus of the body of this thesis. That is the $\delta$-state \CRDT, as described
in~\citet{almedia18}. In their original work, Almeida and his co-authors
describe $\delta$-state \CRDTs as:
\begin{quote}
...ship[ping] a \textit{representation of the effect} of recent update operations
on the state, rather than the whole state, while preserving the idempotent
nature of \textit{join}.
\end{quote}

We will present an example of the $\delta$-state \CRDT in a below section. For
now, we focus on the background material necessary to contextualize
$\delta$-state \CRDTs. This refinement can be thought of as taking ideas from
both state- and op-based \CRDTs to mediate some of the trade-offs described
above. Like a state-based \CRDT, $\delta$-state based \CRDTs have both internal
states and message payloads that form a join semi-lattice. This endows the
$\delta$-state \CRDT with a commutative, associative, and idempotent \emph{join}
operator, as before. Likewise, this means that the $\delta$-state \CRDT supports
relaxed delivery semantics, such as delayed, dropped,\footnote{In this thesis,
we consider dropped messages as having been delayed for an infinite amount of
time, allowing us to reason about a smaller set of delivery semantics.}
reordered, and duplicated message delivery.

Unlike a state-based \CRDT, however, $\delta$-state \CRDTs do not send their
internal state $s^k$ after an update at time $k-1$. We require that these states
have natural representations of their \emph{updates} which do not require
sending the full state to all other replicas. In many circumstances, these
updates can often be represented as ``smaller'' items within the set of all
possible reachable states. For example, in a \CRDT which supports adding to a
set of items, a $\delta$-mutation may be the singleton set containing the
newly-added item, whereas a traditional state-based \CRDT may include the full
set.

This means that:
\begin{itemize}
  \item $\delta$-state \CRDTs support the same weak requirements from the
    network as ordinary state-based \CRDTs. That is, they support dropping,
    duplicating, reordering, and delaying of messages.
  \item $\delta$-state \CRDTs have similarly low-overhead of message size as
    op-based \CRDTs.
\end{itemize}
On the converse, $\delta$-state \CRDTs do not:
\begin{itemize}
  \item ...have potentially large payload size, as state-based \CRDTs are prone
    to have.
  \item ...require a strong delivery semantics $P$ that ensures ordered,
    at-most-once delivery as op-based \CRDTs do.
\end{itemize}
Said otherwise, $\delta$-state \CRDTs have the relative strengths of both state-
and op-based \CRDTs without their respective drawbacks. This makes them an area
of interest, and they are the subject to which we dedicate the remainder of this
thesis.

  \chapter{Elementary \CRDT instantiations}
\label{chap:crdt-instantiations}

In this chapter, we provide the specification of two common \CRDT instantiations
in an op-, state-, and $\delta$-state based style. We discuss the Grow-Only
Counter (G-Counter) and Grow-Only Set (G-Set). Both of these will be the subject
of our verification efforts in Chapter~\ref{chap:example-crdts}.

In each of the below, we assume that $\mathcal{I}$ refers to the set of node
identifiers corresponding to the active replicas. In this thesis, we
consider $\mathcal{I}$ to be fixed during execution; that is, we do not support
addition or deletion of replicas. In practice, \CRDTs do support a dynamic set
of replicas, but we make this assumption for the simplicity of our formalism.

\section{Example: Grow-Only Counter}
\label{sec:example-gcounter}

\subsection{State-based G-Counter}
The G-Counter supports two very simple operations: \textsf{inc} (increment), and
query. When \textsf{inc} is invoked, the counter updates its internal state to
increment the queried value by one. When query is invoked, the counter returns a
number which represents the number of increment operations that have occurred
globally in the system, for which the replica processing the query knows about.
Note that this number is always at least as large as the number of times that
\textsf{inc} has been invoked \textit{at that replica}, and never larger than
the true value of times \textsf{inc} has been invoked globally.

This is our first example of \SEC, where replicas that are ``behind,'' i.e.,
that have not received all updates from all other replicas, are not guaranteed
to reflect the same value upon being queried.\footnote{Perhaps these messages
were delayed or dropped in transit, or otherwise the other replicas have not
broadcast their updates yet. The latter is uncommon in traditional state-based
\CRDTs, but is an often-used operation in variants of state-based \CRDTs
(including $\delta$-state \CRDTs) where updates are bundled into
\emph{intervals} which are sent in a way that preserves causality of updates.}
Concretely, suppose that an \textsf{inc} has occurred at at least one other
replica which has not yet broadcast its updated state. The replica being queried
will have therefore not yet merged the updated state from the replica(s)
receiving \textsf{inc},\footnote{Because we cannot merge updates we do not know
about.} and so those update(s) will not be reflected in the value returned by
querying.

We present a state-based G-Counter \CRDT for concreteness, and then discuss its
definition:

\begin{figure}[H]
  \centering
  \[
    \textsf{G-Counter}_s = \left\{\begin{aligned}
      S &: \mathbb{N}_0^{|\mathcal{I}|} \\
      s^0 &: \left[ 0, 0, \cdots, 0 \right] \\
      q &: \lambda s.\, \sum_{i \in \mathcal{I}} s(i) \\
      u &: \lambda s,i.\, s\left\{ i \mapsto s(i) + 1 \right\} \\
      m &: \lambda s_1, s_2.\, \left[ \max\left\{ s_1(i), s_2(i) \right\}: i \in \mathsf{dom}(s_1) \cup
      \mathsf{dom}(s_2) \right]
    \end{aligned}\right.
  \]
  \caption{Specification of a state-based \textsf{G-Counter} \CRDT.}
  \label{fig:state-gcounter}
\end{figure}

Notice that the state space $\mathbb{N}^{|\mathcal{I}|}_0$ does not match the
return type of the query function, $q$, which is simply $\mathbb{N}_0$. In
Figures~\ref{fig:state-gcounter} and~\ref{fig:state-gcounter-vec-example}, we
utilize a \emph{vector} counter, which should be familiar to readers acquainted
with \emph{vector clocks}~\citep{lamport78}.\footnote{Unlike traditional vector
clocks, the vector \emph{counter} only stores in each replica's slot the number
of \textsf{inc} operations performed \emph{at that replica}.}

\begin{figure}[H]
  \centering
  \includegraphics[width=.6\textwidth]{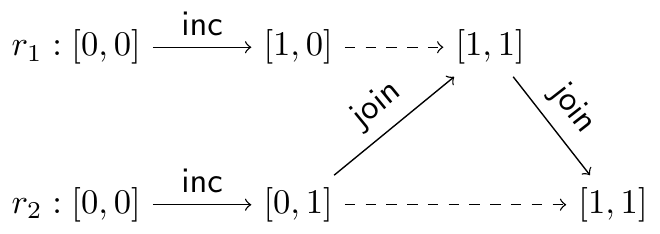}
  \caption{A correct execution of vector-based state G-Counters exchanging
    updates.}
  \label{fig:state-gcounter-vec-example}
\end{figure}

When an \textsf{inc} is invoked at the $i$th replica, it updates its own state
to increment by one the vector element associated with the $i$th replica, here
denoted $s\{i \mapsto s(i) + 1\}$. Finally, upon receiving an update from
another replica, the pair-wise maximum is taken on each of the vector elements.
Note that this is a commutative, associative, and idempotent operation, and so
it forms the least upper-bound of a lattice of vectors of natural numbers.

\subsection{op-based G-Counter}

In the op-based variant of the G-Counter, we can
rely on a delivery semantics $P$ which guarantees at-most-once message
delivery.\footnote{That is, the network is allowed to drop, reorder, and delay
messages, but a single message will never be delivered more than once.} From
this, we say that replicas which are ``behind'' have not yet received the set of
all \textsf{inc} operations performed at other replicas. Replicas which are
``behind'' may ``catch up'' when they receive the set of undelivered messages.
However, these replicas never are ``ahead'' of any other replica, i.e., they
never receive a message which doesn't correspond to a single \textsf{inc}
operation at some other replica, thus they need not be idempotent.

We present now the full definition of the op-based G-Counter:

\begin{figure}[H]
  \centering
  \[
    \textsf{G-Counter}_o = \left\{\begin{aligned}
      S &: \mathbb{N}_0 \\
      s^0 &: 0 \\
      q &: \lambda s.\, s \\
      t &: \textsf{inc} \\
      u &: \lambda s,p.\, s + 1 \\
    \end{aligned}\right.
  \]
  \caption{Specification of an op-based \textsf{G-Counter} \CRDT.}
\end{figure}

Because replicas are sometimes behind but never ahead, we know that the number
of messages received at any given replica is no greater than the sum of the
number of \textsf{inc} operations performed at other replicas, and the number of
\textsf{inc} operations performed locally. So, the op-based G-Counter needs only
to keep track of the number of \textsf{inc} operations it knows about globally,
and this can be done using a single natural number. Hence, $S = \mathbb{N}_0$,
and the bottom state is $0$.

The query operation $q$ is as straightforward as returning the current state.
The \emph{prepare-update} function $t$ always produces the sentinel
\textsf{inc}, indicating that an increment operation should be performed at the
receiving replica. Finally, $u$ takes a state and an arbitrary
payload\footnote{Unused in the implementation here, since the only operation is
\textsf{inc}.} and returns the successor.

Another approach to specifying the op-based G-Counter \CRDT would be to more
closely mirror the state-space of its state-based counterpart, as follows:

\begin{figure}[H]
  \centering
  \[
    \textsf{G-Counter}_o' = \left\{\begin{aligned}
      S &: \mathbb{N}_0^{|\mathcal{I}|} \\
      s^0 &: [ 0, 0, \cdots, 0 ] \\
      q &: \lambda s.\, \sum_{i \in \mathcal{I}} s(i) \\
      t &: (\textsf{inc}, i)) \\
      u &: \lambda s,p.\, s\{ i \mapsto s(i) + 1 \} \\
    \end{aligned}\right.
  \]
  \caption{Alternative specification of an op-based \textsf{G-Counter} \CRDT.}
\end{figure}

where $i$ represents the local node's identifier. Note that, while correct,
restrictive delivery semantics $P$ do not require such a verbose specification,
since the at-most-once delivery guarantees allow us to simply increment our
local count each time we receive an update, since no updates are duplicated over
the network.

\subsection{$\delta$-state based G-Counter}
We conclude this subsection by turning our attention to the $\delta$-state based
G-Counter. We begin first by presenting its full definition:

\begin{figure}[H]
  \centering
  \[
    \textsf{G-Counter}_\delta = \left\{\begin{aligned}
      S &: \mathbb{N}_0^{|\mathcal{I}|} \\
      s^0 &: \left[ 0, 0, \cdots, 0 \right] \\
      q^\delta &: \lambda s.\, \sum_{i \in \mathcal{I}} s(i) \\
      u^\delta &: \lambda s,i.\, \left\{ i \mapsto s(i) + 1 \right\} \\
      m^\delta &: \lambda s_1, s_2.\, \left\{ \max\left\{ s_1(i), s_2(i) \right\}: i \in \mathsf{dom}(s_1) \cup
      \mathsf{dom}(s_2) \right\}
    \end{aligned}\right.
  \]
  \caption{Specification of a $\delta$-state based \textsf{G-Counter} \CRDT.}
\end{figure}

It is worth mentioning the extreme levels of similarity it shares with its
state-based counterpart. Like the state-based G-Counter, the $\delta$-state
based G-Counter uses the state-space $\mathbb{N}^{|\mathcal{I}|}_0$, and has
$s^0 = [0, 0, \cdots, 0]$. Its query operation and merge are defined
identically.

However, unlike the state-based G-Counter, the $\delta$-state based G-Counter
implements the update function as $\lambda s,i.\, \{ i \mapsto s(i) + 1\}$. That
is, instead of returning the amended map (recall: $s\{ \cdots \}$), the
$\delta$-state based G-Counter returns the \emph{singleton map} containing
\emph{only} the updated index. Because of the definition of $m$ (namely, that it
does a pairwise maximum over the \emph{union} of the domains of the two states),
sending the singleton map is equivalent to sending the full map with all other
entries being equal.

\begin{figure}[H]
  \centering
  \includegraphics[width=.7\textwidth]{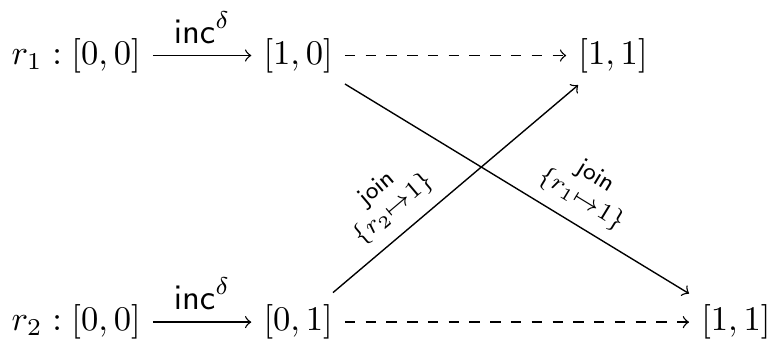}
  \caption{A pair of vector-based $\delta$-state G-Counters replicas exchanging
    updates with each other.}
\end{figure}

This follows from the facts that: (1) the entry being updated has the same
pairwise maximum independent of all other entries in the map, and (2) the
pairwise maximum of all \emph{other} entries does not depend on the updated
entry. So, taking the pairwise maximum of any state with the singleton map
containing one updated value is equivalent to taking the pairwise maximum with
our own state modulo one updated value. $m$ is therefore referred to as a
\emph{$\delta$-mutator}, and the value it returns is an \emph{$\delta$
mutation}~\citep{almedia18}.

This principle of sending \emph{smaller} states (the $\delta$ mutations) which
communicate only the \emph{changed} information is a general principle which
we will return to in the remaining example.

\section{Example: G-Set}
\label{sec:example-gset}

The G-Set is the other primitive \CRDT that we study in this thesis. In essence,
the G-Set is a \emph{monotonic set}. In other words, the G-Set supports the
insertion and query operations, but does not support item removal. This
is a natural consequence of the state needing to form a monotone semi-lattice,
where set deletion would destroy the lattice structure.\footnote{To support
removal from a \CRDT-backed set, the 2P-Set is often used. Verifying this object
is left to future work, which we discuss in
Section~\ref{sec:future-pair-locale}.}

\subsection{State-based G-Set}

We begin our discussion with the state-based G-Set \CRDT, the definition of
which we present below. This is our first example of a \emph{parametric} \CRDT
instance, where the type of the \CRDT is defined in terms of the underlying set
of items that it supports.

\begin{figure}[H]
  \centering
  \[
    \textsf{G-Set}_s(\mathcal{X}) = \left\{\begin{aligned}
      S &: \mathcal{P}(\mathcal{X}) \\
      s^0 &: \{ \} \\
      q &: \lambda x.\, x \in s \\
      u &: \lambda x.\, s \cup \{ x \} \\
      m &: \lambda s_1, s_2.\, s_1 \cup s_2 \\
    \end{aligned}\right.
  \]
  \caption{state-based \textsf{G-Set} \CRDT}
  \label{fig:gset-state}
\end{figure}

For some set $\mathcal{X}$, we can consider the state-based G-Set \CRDT
instantiated over it, $\textsf{G-Set}_s(\mathcal{X})$. The state-space of this
\CRDT is the power set of $\mathcal{X}$, which we denote
$\mathcal{P}(\mathcal{X})$. Initially, the G-Set begins as the empty set, here
denoted $\{ \}$. The three operations are defined as follows:
\begin{itemize}
  \item The query function $q$ is an unary relation, i.e., it determines which
    elements are contained in the G-Set.
  \item The update function $u$ produces the updated set formed by taking the
    union of the existing set, and the singleton set containing the item
    to-be-added.
  \item Finally, the merge function $m$ takes the union of two sets.
\end{itemize}
Note crucially that the merge function $\cup$ defines the least upper-bound of
two sets, and thus endows our \CRDT with a lattice structure. In this lattice of
sets, we say that for some set $\mathcal{X}$, the lattice formed is $\langle
\mathcal{P}(\mathcal{X}), \subseteq \rangle$.

\subsection{op-based G-Set}
In the op-based variant of the G-Set \CRDT, we replace the state-based \CRDT's
update function $u$ with the op-based pair $(t,u)$. The state space, initial
state, as well as the query and merge functions ($q$ and $m$, respectively) are
defined identically. We present the full definition as follows:

\begin{figure}[H]
  \centering
  \[
    \textsf{G-Set}_o(\mathcal{X}) = \left\{\begin{aligned}
      S &: \mathcal{P}(\mathcal{X}) \\
      s^0 &: \{ \} \\
      q &: \lambda x.\, x \in s \\
      t &: \lambda x.\, (\textsf{ins}, x) \\
      u &: \lambda p.\, s \cup \{(\textsf{snd}~p)\} \\
      m &: \lambda s_1, s_2.\, s_1 \cup s_2 \\
    \end{aligned}\right.
  \]
  \caption{op-based \textsf{G-Set} \CRDT}
\end{figure}

The only difference between this \CRDT instantiation and the state-based one is
in the definition of $(t,u)$.\footnote{This is a pattern that will become
familiar during Chapter~\ref{chap:example-crdts}.} In the state-based \CRDT, we
sent the updated state, i.e., $s \cup \{ x \}$. In the op-based variant, we send
a \emph{representation} of the effect, which we take to be the pair
$(\textsf{ins}, x)$, where $\textsf{ins}$ is a sentinel marker indicating that
the second element in the pair should be inserted.

Upon receipt of the message $(\textsf{ins}, x)$, our op-based G-Set \CRDT
computes the new state $s \cup \{ (\snd~p) \}$, where $p$ is the message
payload.

\subsection{$\delta$-state based G-Set}
Finally, we turn our attention to the $\delta$-state based G-Set \CRDT. As was
the case with the $\delta$-state based G-Counter \CRDT, this object is defined
identically as to the state-based counter, with the notable exception of its
update function, $u$.\footnote{This again will be another familiar pattern in
Chapter~\ref{chap:example-crdts}.}

For full formality, we present its definition below:

\begin{figure}[H]
  \centering
  \[
    \textsf{G-Set}_\delta(\mathcal{X}) = \left\{\begin{aligned}
      S &: \mathcal{P}(\mathcal{X}) \\
      s^0 &: \{ \} \\
      q^\delta &: \lambda x.\, x \in s \\
      u^\delta &: \lambda x.\, \{ x \} \\
      m^\delta &: \lambda s_1, s_2.\, s_1 \cup s_2 \\
    \end{aligned}\right.
  \]
  \caption{$\delta$-state based \textsf{G-Set} \CRDT}
\end{figure}

Here, the only difference is between the state- and $\delta$-state based \CRDT's
definition of the update method, $u$. In the state-based G-Set, update was
defined as $u : \lambda x.\, s \cup \{ x \}$. But in the $\delta$-state based
G-Set, the update is defined as $u : \lambda x.\, \{ x \}$. Note crucially that
these two kinds of updates are equal when applied to the same local state.
Consider a state- and $\delta$-state based G-Set, both starting at the same
state $s^t$. For the state-based G-Set, we have:
\[
  \begin{aligned}
    m(s^t, u(x))
      &= s^t \cup (s^t \cup \{ x \}) \\
      &= (s^t \cup s^t) \cup \{ x \} \\
      &= s^t \cup \{ x \} \\
  \end{aligned}
\]
whereas for the $\delta$-based G-Set, we have directly:
\[
  m^\delta(s^t, u^\delta(x)) = s^t \cup \{ x \}
\]

  \chapter{\CRDT reductions}
\label{chap:crdt-reductions}

This chapter outlines the key component of our proof strategy. We begin with a
reduction allowing us to convert from state- to op-based \CRDTs. This reduction
is used in Chapter~\ref{chap:example-crdts} to show a preliminary encoding of
two state-based \CRDTs. We conclude with a reduction from $\delta$-state to
op-based \CRDTs, which is used extensively in the latter part of
Chapter~\ref{chap:example-crdts} to show that $\delta$-state \CRDTs achieve
\SEC.

Specifically, we will discuss the following:
\begin{itemize}
  \item In Section~\ref{sec:state-as-op}, we will describe a mapping
    $\phi_{\text{state} \to \text{op}}$ to reduce state-based \CRDTs to op-based
    \CRDTs.
  \item In Section~\ref{sec:delta-as-op}, we will describe a mapping
    $\phi_{\delta \to \text{op}}$ to reduce $\delta$-state \CRDTs to op-based
    \CRDTs.
\end{itemize}

We state these reductions as ``maxims''. They are stated here in brief, but we
will return to them in Sections~\ref{sec:state-as-op} and~\ref{sec:delta-as-op}.

\begin{maxim}
  \label{maxim:state-as-op}
  A state-based \CRDT is an op-based \CRDT where the \emph{prepare-update} phase
  returns the updated state, and the \emph{effect-update} is a join of two
  states.
\end{maxim}

\begin{maxim}
  \label{maxim:delta-as-op}
  A $\delta$-state based \CRDT is an op-based \CRDT whose messages are
  $\delta$-fragments, and whose operation is a pseudo-join between the current
  state and the $\delta$ fragment.
\end{maxim}

\section{state-based \CRDTs as op-based}
\label{sec:state-as-op}

This section describes a reduction from state-based \CRDTs to op-based \CRDTs.
We describe this reduction to exemplify how to reduce between \CRDT classes, and
use this in Chapter~\ref{chap:example-crdts} to show that two state-based \CRDTs
achieve \SEC.

Consider some state-based \CRDT $C = (S, s^0, q, u, m)$. This object $C$ has a
set of states $S$, an initial state $s^0$, along with functions for querying the
state ($q$), updating its state ($u$), and merging its state with the state of
some other object $(m)$. Our question is to define a mapping $\phi$ as
follows:
\[
  \phi_{\text{state} \to \text{op}} :
    \underbrace{(S, s^0, q, u, m)}_{\text{state-based \CRDTs}} \longrightarrow
    \underbrace{(S, s^0, q, t, u, P)}_{\text{op-based \CRDTs}}
\]
For our purposes, we view $\phi_{\text{state} \to \text{op}}$ as a homomorphism
between state- and op-based \CRDTs.

Note that $P$ (the delivery precondition on the right-hand side) is the only
element which does not have a natural analog on the left-hand side.
Traditionally it is common to have a $P$ which preserves causality, but this is
not necessary for our proofs (since we map states identically as in the
following section). Therefore, we assume that $P$ is always met, in which case
delivery can always occur immediately on the right-hand side.

We'll now turn to describing the details of $\phi_{\text{state} \to \text{op}}$,
which for convenience in this section, we'll abbreviate as simply
$\phi$.\footnote{In the following section, we'll define a new homomorphism
between op- and $\delta$-state based \CRDTs, at which point we will
distinguish between the two mappings when it is unclear which is being referred
to.} To understand $\phi$, we'll consider how it maps the state $S$ (along with
$s^0$ and $q$) separately from how it maps the update procedure $u$.

\subsection{Mapping states under $\phi$}

Let us begin our discussion with a consideration to how $\phi$ maps the state
$S$ from a state-based \CRDT to an op-based one. In practice, it would be
unrealistic to treat the state space of a state-based \CRDT as equal to that of
its op-based counterpart. Doing so would discard one of the key benefits of
op-based \CRDTs over state-based ones, which is that they are often able to
represent the same set of query-able states using simpler structures. For
example, state-based counters (such as the $\textsf{G-Counter}_s$ and
$\textsf{PN-Counter}_s$) often use a vector representation to represent the
number of ``increment'' operations at each node, but op-based counters often
instead use scalars (cf., the examples in Section~\ref{sec:example-gcounter}).

In order to make $\phi$ a simple reduction, we allow the state spaces of the
\CRDT before and after the reduction to be identical. Though \CRDT designers can
often be more clever than this in practice, this makes reasoning about the
transformation much simpler for the purposes of our proofs. Likewise, since the
query function $q$ is defined in terms of the state-space, $S$, we let $\phi$
preserve the implementation of $q$ under the mapping, too.

\subsection{Mapping updates under $\phi$}
\label{sec:states-under-phi}
Now that we have described the process by which $\phi$ maps $S$, $s^0$, and $q$,
we still need to address the implementation of $u$ and $m$ under mapping. Our
guiding principle is the following theorem (which we state and discuss here, but
have not mechanized):
\begin{theorem}
  Let $C_s$ be a state-based \CRDT with $C_s = (S, s^0, q, u, m)$. Define an
  op-based \CRDT $C_o$ as follows:
  \[
    C_0 = \left\{ \begin{aligned}
      S_o &: S \\
      s^0_o &: s^0 \\
      q_o &: q \\
      t_o &: \lambda p.\, u(p...) \\
      u_o &: \lambda s_2.\, m(s^t, s_2) \\
    \end{aligned} \right.
  \]
  then, $C_o$ and $C_s$ reach equivalent states given equivalent updates and
  delivery semantics.
\end{theorem}
\begin{proof}[Proof sketch]
By simulation. Since $s^0_o = s^0$, both objects begin in the same state. Since
$q_o = q$, if the state of $C_o$ and $C_s$ are equal, then $q_o$ will reflect as
much. Finally, an update is prepared locally by computing the updated
state-based representation. That update is applied both locally and at all
replicas by merging the prepared state into $C_o$'s own state, preserving the
equality.
\end{proof}

In other words, we decompose the \textit{update} function of a state-based \CRDT
into the \textit{prepare-update} and \textit{effect-update} functions of an
op-based \CRDT. Let $p$ be the set of parameters used to invoke the update
function $u$ of a state-based \CRDT, i.e., that $u(p...)$ produces the desired
updated state. Then the \textit{prepare-update} returns a serialized
representation of $u(p...)$, which is to say that it returns the updated state.
The \textit{effect-update} implementation then takes that representation and
applies it by invoking the merge function $m$ with the effect representation and
its own state to produce the new state.

This introduces Maxim~\ref{maxim:state-as-op}, which unifies state- and op-based
\CRDTs as behaving identically when the op-based \CRDT performs a join of two
states. We restate this Maxim for clarity:

\setcounter{maxim}{0}
\begin{maxim}
  A state-based \CRDT is an op-based \CRDT where the \emph{prepare-update} phase
  returns the updated state, and the \emph{effect-update} is a join of two
  states.
\end{maxim}

\section{$\delta$-state based \CRDTs as op-based}
\label{sec:delta-as-op}

In the previous section, we described a general procedure for converting
state-based \CRDTs into op-based \CRDTs. In this section, we treat the insight
from the previous section as guidance for how to design a similar reduction to
convert $\delta$-state \CRDTs into op-based \CRDTs. We will use this reduction
to encode $\delta$-state \CRDTs into the library presented in~\citet{gomes17} in
order to verify that $\delta$-state \CRDTs are \SEC.

Similarly as in the previous section, we describe a (new) mapping $\phi_{\delta
\to \text{op}}$ of type:
\[
  \phi_{\delta \to \text{op}} :
    \underbrace{(S, s^0, q, u^\delta, m^\delta)}_{\text{$\delta$-based \CRDTs}}
    \longrightarrow
    \underbrace{(S, s^0, q, t, u, P)}_{\text{op-based \CRDTs}}
\]

For the same reasons as in Section~\ref{sec:states-under-phi}, we let $\phi$
preserve the state space, initial state, and query function. Again, we let $P$
be the delivery precondition which is always met (since messages exchanged are
idempotent, and so there is no need to preserve either causality or at-most-once
delivery as is traditional).

In Section~\ref{sec:state-as-op}, we treated a state-based \CRDT's state as
the representation of the effect for an operation-based \CRDT. In this section,
we do the same for the \emph{$\delta$-state fragment}, which we naturally think
as a difference of two states.

Concretely, let $t : S \to S \to T$ for some type $T$ not necessarily equal to
$S$ which represents the type of all $\delta$-fragments. We define two examples
as follows:
\begin{itemize}
  \item For the G-Set \CRDT, the $\delta$ mutator, $m^\delta$ produces the
    singleton set containing the element added in the last operation. Since only
    one item can be added at a time, computing the following with the before-
    and after-states is sufficient to generate the representation:
    \[
      t = \lambda s_1, s_2.\, s_2 \setminus s_1
    \]
    where $S = T = \mathcal{P}(\mathcal{X})$. This is an example where the \CRDT
    has a type where both the state- and $\delta$-state fragments are members of
    $S$.
  \item For the G-Counter \CRDT, the $\delta$ mutator produces a pair
    type containing the identifier of a node with a changed value, and the new
    value which is assigned to that identifier. $t$ is defined as:
    \[
      t = \lambda s_1, s_2.\, \min_{\substack{i \in \mathcal{I} \\ s_1[i] \ne
        s_2[i]}} (i, s_2[i])
    \]
    (Observe that this function is not defined for two states $s_1 = s_2$, nor
    does it need to be, since the before- and after states are guaranteed to be
    different after invoking $u^\delta$).

    Here we have an example of $T \ne S$, where $T$ instead equals
    $\isacharprime\isa{id} \times \mathbb{N}$.
\end{itemize}

$t$ is now capable of generating the $\delta$-state fragment corresponding to
any pair of states from before and after and invocation of $u^\delta$. Now we
need to define the op-based \CRDT's implementation of $u$ to recover a new state
given a value of type $T$. Here, let $u : S \to T \to S$, which takes in a
current state as well as a $\delta$-fragment and produces a new state.

Intuitively, $u$ is a sort of inverse over the last argument and return value of
$t$. That is, where $t$ was taking the difference of two states, $u$ recovers
that difference into a new state. We define two example implementations of $u$
as follows:
\begin{itemize}
  \item For the G-Set \CRDT, the new state is recovered by taking the union of
    the current state, along with the state carrying the new item. That is:
    \[
      u = \lambda s, t.\, s \cup t
    \]
  \item For the G-Counter \CRDT, the new state is recovered by taking the old
    state, and replacing the entry whose index is equal to the first part of an
    update with the value described by the second part of that update.
\end{itemize}

Importantly, $u$ and $t$ needs to satisfy three important properties:
\begin{enumerate}
  \item $u$ and $t$ can never work together to produce a state which is not
    by either the current state, or the $\delta$-fragment. That is, for any
    state $s'$, we must have that:
    \[
      \forall s \sqsubseteq s'.\, u(s, t(s, s')) \sqsubseteq s'
    \]
    Or in other words, if our starting state is lower in the lattice than $s'$,
    taking the $\delta$-fragment between $s$ and $s'$ and then re-applying that
    to $s$ cannot produce a new state which is $\sqsupseteq s'$.
  \item At all times, all replicas must reflect all updates performed at that
    replica.
  \item All replicas which have received the same set of messages have the same
    state.
\end{enumerate}

Together, these properties are sufficient to re-introduce
Maxim~\ref{maxim:delta-as-op}, which we restate here for clarity:
\setcounter{maxim}{1}
\begin{maxim}
  A $\delta$-state based \CRDT is an op-based \CRDT whose messages are
  $\delta$-fragments, and whose operation is a pseudo-join between the current
  state, and the $\delta$ fragment.
\end{maxim}

Therefore, we have a straightforward procedure for reasoning about
$\delta$-state \CRDTs in terms of op-based \CRDTs, which is to convert any
$\delta$-state \CRDT into an op-based \CRDT, and then use the existing framework
of~\citet{gomes17} to mechanize that that \CRDT achieves \SEC.

  \chapter{Example \CRDTs under Relaxed Network Model}
\label{chap:example-crdts}

We have mechanized proofs that two state- and $\delta$-state based \CRDTs
achieve \SEC. We relax the underlying network model to support non-unique
messages (Section~\ref{sec:network-relaxations}), and then showed that both the
state- and $\delta$-state based G-Counter and G-Set inhabit \SEC
(Sections~\ref{sec:isabelle-state-crdts} and~\ref{sec:isabelle-delta-crdts}).
Finally, we present an alternative encoding of the reduction in
Chapter~\ref{chap:crdt-reductions} for $\delta$-state \CRDTs
(Section~\ref{sec:alternate-delta-encoding}).

\section{Network Relaxations}
\label{sec:network-relaxations}
In~\citet{gomes17}, Gomes and his co-authors provided a network model which
makes the following set of assumptions:
\begin{enumerate}
  \item All messages received by some node were broadcast by some other node.
  \item All messages broadcast by some node were received by that node (i.e.,
    all messages are delivered locally in a reliable fashion).
  \item All messages are unique.
\end{enumerate}
These assumptions allow the network to drop, reorder, and delay messages in
transit.

Because op-based \CRDTs only deliver updates once, it is traditional to assume a
delivery relation $P$ which predicates the set of network executions that we are
allowed to reason about. For example, a network execution which drops all
messages in transit, or does not preserve causality cannot be shown to exhibit
\SEC, and so it is not a member of the relation $P$.  Such an assumption is
standard in the literature and goes back to the original work
in~\citet{shapiro11}.

In~\citet{gomes17}, the authors make extensive use of Isabelle's \emph{locale}
feature~\citep{wenzel02}, which for our purposes we can consider as Isabelle's
implementation of parametric proofs. Specifically, \citet{gomes17} define a
locale for \SEC, which they call
$\isa{strong}\isacharunderscore\isa{eventual}\isacharunderscore\isa{consistency}$.
To instantiate this locale, \CRDT replicas must meet the following
preconditions:

\begin{itemize}
  \item Messages which have a causal dependence are delivered in-order;
    concurrent messages may be delivered in any order (i.e., the $\prec$
    relation is preserved during delivery).
  \item The set of messages delivered at each node is distinct.\footnote{Note
    that the messages transited by the network may be non-distinct. This is
    another standard assumption which can be implemented by tagging each message
    with a vector clock or assigning a globally unique identifier, and having
    each receiving node discard duplicates.}
  \item That concurrent operations commute.
  \item That correct nodes do not fail, i.e., that they remain responsive during
    the execution.
\end{itemize}

While we consider the above to be a reasonable delivery semantics, we wish to
relax the network model in order to support duplicated messages. This
behavior is not permitted by the original network model in~\citet{gomes17},
which assumes that each message in transit on the network has a unique
identifier.

To see this, consider the following example:

\begin{example}
  \label{example:state-op-dup-msgs}
  Consider two systems which have multiple replicas of \CRDT counters. System
  $A$ uses op-based counters, and system $B$ uses state-based counters. Consider
  two replicas in each system, call these $r_1$ and $r_2$. Suppose the following
  happens in each system:

  \begin{itemize}
    \item A \textsf{inc} operation is performed at replica $r_1$, which causes a
      message to be sent to all other replicas. In system $A$, this message is
      $[1, 0, \cdots, 0]$, and in system $B$ this messages is $\textsf{inc}$.
    \item While in route to replica $r_2$, this message is duplicated, and both
      copies are received at replica $r_2$.
  \end{itemize}

  Notice that $q(r_2)$ results in a different value based on whether or not you
  queried the replica belonging to system $A$ or system $B$. In system $A$, the
  duplicate message is ``ignored,'' since merging the same message twice is
  idempotent due to $\sqcup$, and $q(r_2) = 1$ as expected. In system $B$, the
  additional update \emph{is} applied, meaning that $q(r_2) = 2$, which is a
  safety violation.
\end{example}

So, while it is often a safety violation for an op-based \CRDT to receive the
same message twice,\footnote{This is the primary reason why it is a standard
assumption of op-based network models to disallow non-unique messages}
state- and $\delta$-state based \CRDTs can and should tolerate this class of
degenerate behaviors.

The general principle is as follows:
\begin{theorem} \label{thm:state-sec-dup}
  State-based \CRDTs exhibit \SEC even when operating in
  a network environment permitting non-unique messages.
\end{theorem}
\begin{proof}
  By induction on the number of times $i$ a message $m'$ is received. When $i =
  1$, the goal is trivially established. When $i > 1$, the idempotency of
  $\sqcup$ shows that:
  \[
    m \sqcup \underbrace{m' \sqcup \cdots \sqcup m'}_{\text{$i-1$ times}} \sqcup~m'
      = m \sqcup m' \sqcup m'
      = m \sqcup m'
  \]
  where the second equality follows from the inductive hypothesis, and the third
  from the fact that $m' \sqcup m' = m'$ by the idempotency of $\sqcup$.
\end{proof}

This result guides our approach as follows: to show a stronger result that uses
Theorem~\ref{thm:state-sec-dup} (i.e., that state- and $\delta$-state based
\CRDTs achieve \SEC no matter how many times), the network model originally
presented in~\citet{gomes17} should be extended to remove the assumption that
message identifiers are unique.

\subsection{Delivery Semantics}
In their original network model, the authors of~\citet{gomes17} use an Isabelle
\emph{locale} in order to parameterize varying instantiations of the network
based on certain assumptions. They provide the following definition for the
Network locale~\citep{gomes17}:
\begin{figure}[H]
\begin{isabelle}
~~~~~~~~\isakeyword{and}\ \=msg{\isacharunderscore}id{\isacharunderscore}unique{\isacharcolon}\ \={\isasymrbrakk}\ \={\isachardoublequoteopen}Broadcast\ m\ {\isasymin}\ set\ {\isacharparenleft}history\ i{\isacharparenright}\ \=\kill
\isacommand{locale}\ network\ {\isacharequal}\ node{\isacharunderscore}histories\ history\\
~~~~\isakeyword{for}\>history\ {\isacharcolon}{\isacharcolon}\ {\isachardoublequoteopen}nat\ {\isasymRightarrow}\ {\isacharprime}msg\ event\ list{\isachardoublequoteclose}\ {\isacharplus}\\
~~~~\isakeyword{fixes}\>msg{\isacharunderscore}id\ {\isacharcolon}{\isacharcolon}\ {\isachardoublequoteopen}{\isacharprime}msg\ {\isasymRightarrow}\ {\isacharprime}msgid{\isachardoublequoteclose}\\
~~~~\isakeyword{assumes}\ delivery{\isacharunderscore}has{\isacharunderscore}a{\isacharunderscore}cause{\isacharcolon}\\
\>\>{\isasymlbrakk}\ {\isachardoublequoteopen}Deliver\ m\ {\isasymin}\ set\ {\isacharparenleft}history\ i{\isacharparenright}\ \>\>{\isasymrbrakk}\ {\isasymLongrightarrow}\ {\isasymexists}j{\isachardot}\ Broadcast\ m\ {\isasymin}\ set\ {\isacharparenleft}history\ j{\isacharparenright}{\isachardoublequoteclose}\\
~~~~~~~~\isakeyword{and}\>deliver{\isacharunderscore}locally{\isacharcolon}\ \>{\isasymlbrakk}\ \>{\isachardoublequoteopen}Broadcast\ m\ {\isasymin}\ set\ {\isacharparenleft}history\ i{\isacharparenright}\ \>{\isasymrbrakk}\ {\isasymLongrightarrow}\  Broadcast\ m\ {\isasymsqsubset}\isactrlsup i\ Deliver\ m{\isachardoublequoteclose}\\
~~~~~~~~\isakeyword{and}\>msg{\isacharunderscore}id{\isacharunderscore}unique{\isacharcolon}\ \>{\isasymlbrakk}\ \>{\isachardoublequoteopen}Broadcast\ m{\isadigit{1}}\ {\isasymin}\ set\ {\isacharparenleft}history\ i{\isacharparenright};\\
\>\>\>Broadcast\ m{\isadigit{2}}\ {\isasymin}\ set\ {\isacharparenleft}history\ j{\isacharparenright};\\
\>\>\>msg{\isacharunderscore}id\ m{\isadigit{1}}\ {\isacharequal}\ msg{\isacharunderscore}id\ m{\isadigit{2}}\ \>{\isasymrbrakk}\ {\isasymLongrightarrow}\ i\ {\isacharequal}\ j\ {\isasymand}\ m{\isadigit{1}}\ {\isacharequal}\ m{\isadigit{2}}{\isachardoublequoteclose}
\end{isabelle}
\centering
\caption{Isabelle specification of the Network locale as given
  in~\citet{gomes17}.}
\end{figure}

In order to extend the network model of Gomes et. al. to support duplicated
messages, we need to remove the assumption
$\isa{msg}\isacharunderscore\isa{id}\isacharunderscore\isa{unique}$, which
allows the enclosed proofs to assume that messages have unique identifiers.
While this assumption is part of the locale, proofs are allowed to assume that
if two messages $m_1$ and $m_2$ with the same identifier (i.e., that
$\isa{msg}\isacharunderscore\isa{id} m_1 = \isa{msg}\isacharunderscore\isa{id}
m_2$) exists in the history of two nodes, that either the two nodes or two
messages are identical.

Although our proofs are still instantiated after fulfilling the qualifier $P$,
we still wish to reason about an expanded set of network executions which
includes message dropping.\footnote{Since op-based \CRDTs require
causality-preserving semantics $P$, we cannot remove the dependence on $P$
without substantial alternation to the library. We leave this to future work,
and discuss it in greater detail in Chapter~\ref{chap:future-work}.}

For our purposes, we begin by specifying a relaxed $\isa{network}$ locale as
follows:
\begin{isabelle}
~~~~~~~~\isakeyword{and}\ \=msg{\isacharunderscore}id{\isacharunderscore}unique{\isacharcolon}\ \={\isasymrbrakk}\ \={\isachardoublequoteopen}Broadcast\ m\ {\isasymin}\ set\ {\isacharparenleft}history\ i{\isacharparenright}\ \=\kill
\isacommand{locale}\ network\ {\isacharequal}\ node{\isacharunderscore}histories\ history\\
~~~~\isakeyword{for}\>history\ {\isacharcolon}{\isacharcolon}\ {\isachardoublequoteopen}nat\ {\isasymRightarrow}\ {\isacharprime}msg\ event\ list{\isachardoublequoteclose}\ {\isacharplus}\\
~~~~\isakeyword{fixes}\>msg{\isacharunderscore}id\ {\isacharcolon}{\isacharcolon}\ {\isachardoublequoteopen}{\isacharprime}msg\ {\isasymRightarrow}\ {\isacharprime}msgid{\isachardoublequoteclose}\\
~~~~\isakeyword{assumes}\ delivery{\isacharunderscore}has{\isacharunderscore}a{\isacharunderscore}cause{\isacharcolon}\\
\>\>{\isasymlbrakk}\ {\isachardoublequoteopen}Deliver\ m\ {\isasymin}\ set\ {\isacharparenleft}history\ i{\isacharparenright}\ \>\>{\isasymrbrakk}\ {\isasymLongrightarrow}\ {\isasymexists}j{\isachardot}\ Broadcast\ m\ {\isasymin}\ set\ {\isacharparenleft}history\ j{\isacharparenright}{\isachardoublequoteclose}\\
~~~~~~~~\isakeyword{and}\>deliver{\isacharunderscore}locally{\isacharcolon}\ \>{\isasymlbrakk}\ \>{\isachardoublequoteopen}Broadcast\ m\ {\isasymin}\ set\ {\isacharparenleft}history\ i{\isacharparenright}\ \>{\isasymrbrakk}\ {\isasymLongrightarrow}\  Broadcast\ m\ {\isasymsqsubset}\isactrlsup i\ Deliver\ m{\isachardoublequoteclose}\\
\end{isabelle}

Removing this assumption immediately invalidates many of the proofs contained
within the $\isa{network}$ locale. These proofs are broken due to a variety of
reasons, ranging from something as simple as referencing a now-missing
assumption, to more complex issues, e.g., a proof which relies on the uniqueness
of delivered messages.

We now describe our strategy for repairing these proofs:
\begin{enumerate}
  \item First, remove the assumption
    $\isa{msg}\isacharunderscore\isa{id}\isacharunderscore\isa{unique}$ from the
    $\isa{network}\isacharunderscore\isa{with}\isacharunderscore\isa{ops}$
    locale, as above.
  \item Identify the set of broken proofs. In each broken proof, do the
    following:
    \begin{enumerate}
      \item Identify the earliest broken proof step.
      \item Delete it and all proof steps following it.
      \item Replace the proof body with the term $\isakeyword{sorry}$.
    \end{enumerate}
  \item In any order, consider a proof which ends with $\isakeyword{sorry}$, and
    repair the proof.
\end{enumerate}

In total, there were four (4) key lemmas which needed repair. These were:
$\isa{hb}\isacharunderscore\isa{antisym}$,
$\isa{hb}\isacharunderscore\isa{has}\isacharunderscore\isa{a}\isacharunderscore\isa{reason}$,
$\isa{hb}\isacharunderscore\isa{cross}\isacharunderscore\isa{node}\isacharunderscore\isa{delivery}$, and
$\isa{hb}\isacharunderscore\isa{broadcast}\isacharunderscore\isa{broadcast}\isacharunderscore\isa{order}$.
After removing the
$\isa{msg}\isacharunderscore\isa{id}\isacharunderscore\isa{unique}$ assumption,
each of the above four proofs were able to be repaired automatically by
Isabelle's proof search procedure $\isakeyword{sledgehammer}$~\citep{wenzel02}.

In each of the \CRDTs that we do verify, we are required to instantiate a lemma
stating:
\[
  \isa{apply}\isacharunderscore\isa{operations}~\isa{xs}~\isacharequal~
  \isa{apply}\isacharunderscore\isa{operations}~\isa{ys}
\]
where $\isa{xs}$ and $\isa{ys}$ are lists of messages delivered to a pair of
replicas by the network. In other words, no matter what messages are delivered
in what order, the two replicas attain the same state. Following the original
proofs provided for op-based \CRDTs in~\citet{gomes17}, our proofs of this lemma
make the standard assumption that:
\[
  \isakeyword{set}~(\isa{node}\isacharunderscore\isa{deliver}\isacharunderscore\isa{messages}~\isa{xs})~\isacharequal~
  \isakeyword{set}~(\isa{node}\isacharunderscore\isa{deliver}\isacharunderscore\isa{messages}~\isa{ys})
\]
Note that although we require that the set of operations delivered at two nodes
is identical in order for those two nodes to attain the same value, we are able
to reason over an expanded set of network behaviors. For example, if some
message $m$ appears in either of the two sets above, we know that it only
appears in that node's history once, by the
$\isa{msg}\isacharunderscore\isa{id}\isacharunderscore\isa{unique}$ assumption.
But without that assumption, we know instead that it appears \emph{at least}
once in each of the node's log of history.

This is a key distinction, since not knowing how many times a message was
delivered to either of the two replicas means that we are able to conclude that
they reach the same state if the same set of messages is delivered \emph{at
least once} to each of the replicas. Said otherwise, it does not matter how many
times a message was delivered at each of two replicas, so long as it was
delivered at least once at both. This allows us to exercise the latter case of
Example~\ref{example:state-op-dup-msgs} using the relaxed network model.

\section{State-based \CRDTs}
\label{sec:isabelle-state-crdts}

Equipped with a relaxed network model, we are now ready to verify two examples
of state-based \CRDTs.

\subsection{State-based G-Counter}
We begin first with the G-Counter, the formal definition of which can be found in
Section~\ref{sec:example-gcounter}. Following our intuition in
Maxim~\ref{maxim:state-as-op}, we define a type to represent the $\isa{state}$ and
$\isa{operation}$ of a state-based G-Counter, presented below:

\begin{figure}[H]
  \input{figures/theories/gcounter-state}
  \caption{Isabelle definitions for $\isa{state}$ and $\isa{operation}$ for a
    state-based G-Counter \CRDT.}
\end{figure}

Here, we let the state be a partial mapping from a transparent
$\isacharprime\isa{id}$ type (the value of which uniquely identifies a replica
in the system) to an $\isa{int}$ which specifies the number of increment
operations performed at that replica. Like in
Section~\ref{sec:example-gcounter}, this defines a vector-like object, where
each slot in the vector corresponds to the number of increment operations
performed at some unique replica in the system. We define this mapping to be
partial, where the $\isa{None}$ value signals that no increments have been
performed at a given node.\footnote{This choice is arbitrary, and could have
easily have been implemented as mapping to $0$ instead.}

Next we define the operation to be a type-level synonym for the
$\isacharprime\isa{id}~\isa{state}$ type. This encodes that operations
\emph{are} states. We interpret that upon receipt of an operation that we
replace our current state with the join of it and the state encoded by the
operation, which is an implementation of Maxim~\ref{maxim:state-as-op}.

Before introducing the interpretation of
$\isa{gcounter}\isacharunderscore\isa{op}$ (which will be responsible for
performing this join operation as described), we look at a few other functions
which are defined to act over this type:

\begin{figure}[H]
  \input{figures/theories/gcounter-misc}
  \caption{Isabelle definitions for state-based G-Counter-related functions.}
  \label{fig:gcounter-option-max}
\end{figure}

The function $\isa{inc}$ specifies (for demonstration purposes) how to increment
the value in a vector for some node. That is, $\isa{inc}$ specifies the
procedure to execute when an increment operation is performed at some replica.
Since our proofs reason purely about transitions of states, and not the external
forces that drive them, this function is never called by our proofs, but merely
left for the reader as a demonstration of how to drive the system.

The other function $\isa{option}\isacharunderscore\isa{max}$ specifies the
pair-wise maximum of two $\isa{int}~\isa{option}$ values. Note that these are
the right-hand side of the mapping in $\isacharprime\isa{id}~\isa{state}$, and
so this function is used to merge the state received from some other replica.
We will prove some additional facts about this function shortly, but for now we
interpret it as taking the maximum of two optional integers, where a present
integer is always preferred over an absent one,\footnote{That is, the maximum of
$\isa{Some}~x$ and $\isa{None}$ is $\isa{Some}~x$.} and the maximum of two
absent integers is $\isa{None}$.

Now that we have a way to interpret the pair-wise maximum of two states which
constitute a join, we can specify our definition of the ``operation'' for a
state-based G-Counter \CRDT. Recall that as in Maxim~\ref{maxim:state-as-op}, we
need to specify an operation which is the join of two states. We present now the
definition as used in our proofs:

\begin{figure}[H]
  \input{figures/theories/gcounter-state-op}
  \caption{Isabelle definition for the ``operation'' of a state-based G-Counter
    \CRDT.}
\end{figure}

Here, we specify a function that produces a partial mapping from an operation
and state to a new state. The function is not total (that is, it \emph{can}
return $\isa{None}$ for some input) to indicate a crash. For our purposes, we do
not specify such a case, and so the function always returns $\isa{Some}$ for any
input.\footnote{Note that this $\isa{None}$ is different from the partial
mapping of the $\isacharprime\isa{id}~\isa{state}$ type, which specifies that
the \emph{count} of increment operations at some replicas may zero. Returning
$\isa{None}$ from $\isa{gcounter}\isacharunderscore\isa{op}$ indicates that
there is no state at all, i.e., a crash has occurred.} Here, the state on the
left-hand side indicates the state that our replica currently has. The
``operation'' so-to-speak is the state at some \emph{other} replica. By encoding
the state from a remote replica in this fashion, we are implicitly saying that
this state should be joined with our current state, and that the result of this
join should replace our current state. So, we return a new state, which is a
function which maps node identifiers to the maximum of the associated value
between our previous state, and the state at some other replica.

For example, if our state in a four-replica system is:
\[
  \{ r_1 : 1,\, r_2 : \bot,\, r_3 : 2,\, r_4 : \bot \}
\]
and the state of some replica is:
\[
  \{ r_1 : \bot,\, r_2 : 1,\, r_3 : 3,\, r_4 : \bot \}
\]
the resulting state is:
\[
  \{ r_1 : 1,\, r_2 : 1,\, r_3 : 3,\, r_4 : \bot \}
\]

In Isabelle, we encode this as a function which forms a closure over the local
and remote states, and then computes the maximum some given node identifier
$\isa{x}$. In practice, this is the lazy equivalent to computing all of the
values up front upon merging.

Now that we have an instantiation of how to modify and merge states (the
equivalent of the $u$ and $m$), it remains to show that this is a suitable
instantiation of the
$\isa{strong}\isacharunderscore\isa{eventual}\isacharunderscore\isa{consistency}$
locale.\footnote{Recall that instating this locale is equivalent to a proof that
the object it is being instantiated with has \SEC.}

A first-try instantiation shows that it is not possible to do so without
additional proofs. Upon inspecting the unmet goals, we can deduce that Isabelle
wants a proof of the commutativity and associativity of
$\isa{option}\isacharunderscore\isa{max}$, the key function used to implement
the merge of two states. We leave the full definitions of these proofs to
Section~\ref{sec:app-gcounter-comm-assoc}; most are able to be completed with
induction and term simplification only, and so are not of great interest to this
section.

Once we have a proof of commutativity and associativity (Isabelle can infer that
$\isa{option}\isacharunderscore\isa{max}$ is idempotent automatically), we then
state an important lemma and corollary, which are as follows:

\begin{figure}[H]
  \input{figures/theories/gcounter-commute}
  \caption{Isabelle proofs that concurrent operations commute in the state-based
    G-Counter.}
\end{figure}

\begin{figure}[H]
  \input{figures/theories/gcounter-convergence}
  \caption{Isabelle proofs that the state-based G-Counter is convergent.}
\end{figure}

This and the above proof establish the following two lemmas:
\begin{itemize}
  \item Operations that have been delivered at some node can be applied in any
    order up to causality and still achieve the same state (there is a more
    general result, since \emph{all} operations on the G-Counter are concurrent,
    but we specialize to showing a more specific case to guide Isabelle's reuse
    of the proof).
  \item Having the same set of operations delivered at any two replicas ensures
    that those replicas are in the same state.
\end{itemize}

The first property is a helpful lemma which is used in internal proofs, but the
second lemma should be familiar to the reader: this is the safety property of
\SEC! Note also that this is the first time that we are seeing our efforts in
relaxing the network model bear fruit. That is, even though the two \emph{sets}
must be equal, we do not make a restriction on the number of times that a
particular message is delivered at either node. This allows us to prove a
stronger result that this \CRDT achieves a consistent result despite the number
of times that a message was (or was not) duplicated.

Finally, once we have shown these two properties, we can instantiate the
$\isa{strong}\isacharunderscore\isa{eventual}\isacharunderscore\isa{consistency}$
locale, which is witness to the fact that this \CRDT object achieves \SEC. We
present the instantiation now, and leave the proof to
Section~\ref{sec:app-gcounter-comm-assoc}:

\begin{figure}[H]
  \input{figures/theories/gcounter-state-sec}
  \caption{Isabelle proof that the state-based G-Counter \CRDT is \SEC.}
\end{figure}

Incidentally by this point, the proof that \SEC is inhabited by our encoding of
the G-Counter \SEC is mostly automatic, up to giving Isabelle some hints about
rewrite and simplification rules that it should apply.

\subsection{State-based G-Set}
\label{sec:state-gset}

Now that we have verified a state-based G-Counter \CRDT, we turn our attention
to the other \CRDT object for study in this thesis. This will be the state-based
G-Set, which is described in detail in Section~\ref{sec:example-gset}. Readers
may notice that the remaining sections in this chapter are shorter and shorter
as we build up and reuse techniques from earlier proofs in later ones.

For now, we begin with an instantiation of the state-based G-Set \CRDT, as
follows:

\begin{figure}[H]
  \input{figures/theories/gset-state}
  \caption{Isabelle types for the state and operations of a state-based G-Set.}
\end{figure}

Like in Figure~\ref{fig:gset-state}, we parameterize the \CRDT on the type of
element in the set, which we denote in Isabelle as $\isacharprime\isa{a}$.
Similar to our offers in the previous sub-section, we define the
$\isa{operation}$ type to be a type-level synonym for the $\isa{state}$ type,
which we interpret in the same way (that is, that receiving an ``operation''
from some other replica is equivalent to being told to merge our state with the
received one, and replace our current state with the result).

Next, we define a simple insertion operation:

\begin{figure}[H]
  \input{figures/theories/gset-state-op}
  \caption{Isabelle definition of the insertion operation for a state-based
    G-Set.}
\end{figure}

Again, we define an operation $\isa{insert}$ for demonstration
purposes.\footnote{Again, our proofs reason about state \emph{transitions}.} Now
that we have a convenience function for generating states that could be used to
drive state transitions within the system, we can instate the interpretation of
an operation at a state-based G-Set \CRDT. This is the second function in the
above Isabelle snippet. Like the state-based G-Counter \CRDT, we map a pair of
$\isacharprime\isa{a}~\isa{operation}$ and $\isacharprime\isa{a}~\isa{state}$ to
a new state of the same type, or $\isa{None}$.\footnote{The existing library
in~\citet{gomes17} requires that this function be a partial mapping, but we do
not specify any behaviors which would cause our node to crash in ordinary
execution here.}

Faithful to the original specification in Figure~\ref{fig:gset-state}, we
interpret the join of two states (that is, two sets of items, one per replica)
as the merge operation.

Because we are using Isabelle's $\isa{set}$ library and its built-in function
$\cup$, we can leverage proofs about built-in Isabelle types, including the fact
that $\cup$ is commutative, associative, and idempotent. Therefore, unlike our
experience in the previous sub-section when specifying the state-based G-Counter
\CRDT, we do not need to prove these facts ourselves.\footnote{Recall that in
this instance, we were using a user-defined function
$\isa{option}\isacharunderscore\isa{max}$, and had an expanded obligation to
prove that this function was commutative and associative; Isabelle inferred
idempotence automatically.}

Aside from some additional proofs which are standard to all of our
instantiations of \CRDTs using the library from~\citet{gomes17}, we can
immediately instantiate the
$\isa{strong}\isacharunderscore\isa{eventual}\isacharunderscore\isa{consistency}$
locale without additional proof. We present the statement of this locale below,
and leave it and the additional proofs about the state-based G-Set to
Section~\ref{sec:app-state-gset}.

\begin{figure}[H]
  \input{figures/theories/gset-state-sec}
  \caption{Isabelle instantiation of the
    $\isa{strong}\isacharunderscore\isa{eventual}\isacharunderscore\isa{consistency}$
    locale for the state-based G-Set.}
\end{figure}

\section{$\delta$-state based \CRDTs}
\label{sec:isabelle-delta-crdts}

We have reached the climax of this chapter in which we now set out to verify
that $\delta$-state based \CRDT equivalents of the G-Counter and G-Set are also
inhabitants of the \SEC locale.

This follows simply from construction and our reductions. Recall the sentiment
of our claim in Maxim~\ref{maxim:delta-as-op} that all $\delta$-state \CRDTs are
themselves like the op-based equivalent of state-based \CRDTs, only with
additional restriction on what states are sent to other replicas. All that
suffices to show is that restricted executions of the datatype--that is, ones in
which only $\delta$-state fragments are sent, and not full state--still inhabit
the \SEC locale.

Recall that, since our proofs reason about state transitions inductively, we
have implicitly covered the case in which only $\delta$-fragments of state are
exchanged between replicas. This is a consequence of our encoding of
$\delta$-state \CRDTs as op-based \CRDTs, and the fact that all $\delta$-based
\CRDT messages are also state-based \CRDT messages.

Since we have verified our \CRDTs as inhabiting the \SEC locale over all
possible operations, we produced proofs for $\delta$-state \CRDTs as a
side-effect of our strategy in Maxims~\ref{maxim:state-as-op}
and~\ref{maxim:delta-as-op}.

We devote the remainder of this section to stating the types of the
operation-producing functions for the $\delta$-based \CRDT equivalents of the
G-Counter and G-Set.

\subsection{$\delta$-state based G-Counter}
\label{sec:isabelle-delta-gcounter}

We begin first with our full definition of the $\delta$-state based G-Counter
\CRDT. Like the state-based variant, we treat the state as a partial mapping
between a transparent node identifier type and an optional value, referring to
the number of increment operations performed locally at that node. Following
Maxim~\ref{maxim:state-as-op}, we treat the operation again as a type-level
synonym for the state.

Similar to our treatment of the state-based G-Counter \CRDT, we encode the state
as a partial mapping from the set of node identifiers to an integer number of
times that an increment operation was performed at the replica belonging to that
node identifier. Likewise, we treat the operation as a type-level synonym for
this definition of the state.

The only difference (besides renaming $\isa{gcounter}\isacharunderscore\isa{op}$
to $\isa{delta}\isacharunderscore\isa{gcounter}\isacharunderscore\isa{op}$) is
that: the function $\isa{update}$ does not ever return a value from the
underlying state which does not belong to the replica being updated. That is, we
return a state which is \emph{only} defined for the single replica being
updated.

The full definition of the updated operation function in Isabelle is as follows:

\begin{figure}[H]
  \input{figures/theories/delta-gcounter}
  \caption{Isabelle definition of the $\delta$-state G-Counter \CRDT.}
\end{figure}

Here, we return a $\delta$-state which is only defined for the single replica
identifier being incremented. That is, we only return a value which is not
$\isa{None}$ for the occurrence when the parameter $\isa{j}$ is bound to a value
which equals $\isa{who}$. When this is met, we increment the value in the state
by one, and return the sum.

Since the body of the $\delta$-based G-Counter \CRDT is the same, and only the
convenience function changed, all other proofs are the same. In
Section~\ref{sec:alternate-delta-encoding}, we discuss an alternative encoding
which limits the kind of messages being sent at the type-level to be restricted
only to $\delta$-fragments.

\subsection{$\delta$-state based G-Set}

Finally, we turn our attention to the remaining \CRDT instance: the G-Set.
Similar to our experience verifying the $\delta$-based G-Counter, the
specification of the \CRDT itself is identical to the original encoding in
Section~\ref{sec:state-gset}, following our intuition in
Maxim~\ref{maxim:delta-as-op}.

For completeness, we present the full instantiation of this type (again leaving
the additional proofs to the Appendix in Section~\ref{sec:app-delta-gset}):

\begin{figure}[H]
  \input{figures/theories/delta-gset}
  \caption{Isabelle definition of the $\delta$-state G-Set \CRDT.}
\end{figure}

Notice that our encoding is identical as in the state-based G-Set example, but
the definition of $\isa{insert}$ has changed. Instead of constructing and
sending the union of the current set and the singleton set containing the item
we wish to add, we construct only the singleton set.

This guides our understanding that if this \CRDT only sends messages that are
able to be generated from the modified $\isa{insert}$ function, that it will
achieve \SEC, and indeed we are able to instantiate the
$\isa{strong}\isacharunderscore\isa{eventual}\isacharunderscore\isa{consistency}$
locale over this type. Because our proofs are inductive over state
\emph{transitions}, we have implicitly proved the case where only
$\delta$-fragments are sent as well.

In the following section, we discuss an alternate encoding which permits a more
direct proof of this fact.

\section{Alternative encoding of the $\delta$-state reduction}
\label{sec:alternate-delta-encoding}

In this Section, we discuss an alternative encoding in Isabelle of
$\delta$-state \CRDTs. Our key insight following
Maxim~\ref{maxim:delta-as-op} is that in a system where the proofs are done
inductively over state transitions, all executions which only exchange
$\delta$-fragments are implicitly verified. That is, since these messages
comprise a subset of the set of messages which are sent by state-based \CRDTs,
our inductive hypothesis still holds, and the result is preserved for
$\delta$-state \CRDTs.

But the key restriction in Maxim~\ref{maxim:delta-as-op} is that
$\delta$-state \CRDTs are ordinarily allowed to send only \emph{fragments} of
their state, not the entire state.\footnote{This restriction does not hold for
certain anti-entropy algorithms which are implemented on top of $\delta$-based
\CRDTs~\citep{almedia18}. This left to future work and discussed briefly in
Section~\ref{sec:direct-delta-proofs}.} For our purposes, we devote the
remainder of this section to exploring how this restriction is encoded at the
type level in our proofs in Isabelle.

The approach that we take here is to let the operation type be a type-level
synonym for a sort of refinement type of the state. Consider for a brief example
the G-Set \CRDT. Here, the full state is $\isacharprime\isa{a}~\isa{set}$, but
the $\delta$-fragments are singleton sets. Ordinarily we would make a type-level
alias from $\isacharprime\isa{a}~\isa{operation}$ to be the same as
$\isacharprime\isa{a}~\isa{state}$, but this is too permissive. Recall that the
operation--for our purposes--is analogous to the kind of the update message sent
between replicas. We want to encode that this can \emph{only} be the singleton
set, not any arbitrary set. To do this, we let the $\isa{operation}$ type be a
single element of $\isacharprime\isa{a}$ type, which we interpret as the
singleton set.

In the following sections, we will consider two examples of this restriction.
Note that we are proving the same thing, so the underlying proof statement is
unchanged. That is, in Section~\ref{sec:isabelle-delta-crdts} we were reasoning
about an inductive hypothesis over \emph{all} possible state transitions. In
this section, we are reasoning about smaller single transitions (e.g., in the
case of a G-Set, adding at most one element in each step), but this is still
sufficient to reason about all possible state transitions.

\subsection{Refined $\delta$-state based G-Counter}
We begin first with the $\delta$-state based G-Counter, and specify it using our
alternate encoding. Recall that in the original specification in
Section~\ref{sec:isabelle-delta-gcounter}, we let the state type be a (partial)
mapping from a transparent node identifier type $\isacharprime\isa{a}$ to
$\isa{nat}$.

Aliasing the operation type to be a type-level synonym for the state allowed our
$\delta$-state \CRDT instantiation to send \emph{any} message, which is too
permissive. Recall that in a $\delta$-state G-Counter, we typically send a
single update, e.g., $\{ r_1 \mapsto n \}$ for a single replica $r_1$
incrementing the number of operations performed up to $n$. We specify this
single-update in Isabelle as follows:

\begin{figure}[H]
  \input{figures/theories/delta-gcounter-refined-state}
  \caption{Isabelle definitions for the $\isa{state}$ and $\isa{operation}$
    types for the restricted $\delta$-based G-Counter.}
\end{figure}

Here, we encode the restriction that a $\delta$-state G-Counter can only send an
update about a single replica by encoding that its operation type is a pair of
a transparent node identifier value and the number of increments performed at
that node.

In the following figure, we present the remainder of the altered definitions to
work around this more restricted $\isa{operation}$ type.

\begin{figure}[H]
  \input{figures/theories/delta-gcounter-refined-ops}
  \caption{Isabelle definitions of the remainder of functions for the restricted
    $\delta$-state G-Counter.}
\end{figure}

In the above, we omit the definition of
$\isa{option}\isacharunderscore\isa{max}$, which is identical to
Figure~\ref{fig:gcounter-option-max}. First, we reimplement $\isa{inc}$ to
return a value of the correct type by constructing a pair of the node being
incremented, and the value that it is being incremented to. Now that this is
done, we update our implementation of
$\isa{delta}\isacharunderscore\isa{gcounter}\isacharunderscore\isa{op}$ to match
the new type. Again, we return a function which takes the pairwise maximum
between the old and new values corresponding to a given node. However, we can no
longer pass the given operation as input to this function, since it does not
have the same type as the state of our \CRDT in this encoding.

To address this, we \emph{convert} the operation into a state by constructing a
state which is only defined for the single node being updated, and returns
$\isa{None}$ for all other values. Once we have this, we can then call it with
an arbitrary node to take its pairwise maximum to generate an updated state.

After specifying the \CRDT using this alternate encoding, we did not have to
update any of our existing proofs developed in
Section~\ref{sec:isabelle-delta-gcounter}, since Isabelle was able to infer the
remainder of facts it needed to recheck our existing proofs.

\subsection{Refined $\delta$-state based G-Set}

In this section, we apply the same techniques to show that an alternate encoding
of the $\delta$-based G-Set still achieves \SEC. As illustrated
in~\ref{sec:alternate-delta-encoding} we replace the definition of the
$\isa{operation}$ type to only allow for restricted, single-element messages
as follows:

\begin{figure}[H]
  \input{figures/theories/delta-gset-refined-state}
  \caption{Isabelle definitions for the $\isa{state}$ and $\isa{operation}$
    types for the restricted $\delta$-based G-Set.}
\end{figure}

First observe that the underlying type for $\isacharprime\isa{a}~\isa{state}$ is
unchanged, but that the new type for $\isacharprime\isa{a}~\isa{operation}$ only
allows a single value to be communicated in messages between two nodes.

Faced with this additional restriction, we update our proofs accordingly.
Following the example in the previous section, we can imagine that our proofs
will need to be updated in two locations:
\begin{enumerate}
  \item The definition of $\isa{insert}$ will become simplified, since we will
    no longer have to refer to the current state when generating the message
    signaling an item has been inserted.
  \item The interpretation of the operation will become slightly more complex,
    since we will have to treat the incoming item \emph{encoded} in the operation
    as a singleton set, and will thus have to do that conversion.
\end{enumerate}

We include the updated definitions of these two functions in Isabelle below:

\begin{figure}[H]
  \input{figures/theories/delta-gset-refined-ops}
  \caption{Isabelle definitions of the remainder of functions for the restricted
    $\delta$-state G-Set.}
\end{figure}

Notice that our insertion operation became dramatically simpler. In fact, the
function is so simple, that it is the identity on its first
parameter. We could have dropped the second parameter from the function
entirely,\footnote{Making its type signature $\isacharprime\isa{a} \Rightarrow
(\isacharprime\isa{a}~\isa{operation})$.} but we leave it there to illustrate
the fact that it \emph{can} be ignored.

This simplification is balanced with a small amount of complexity added in the
$\isa{delta}\isacharunderscore\isa{gset}\isacharunderscore\isa{op}$ function,
which now constructs a singleton set from the incoming operation--referred
to as $\isa{a}$--into $\{~\isa{a}~\}$.

As before, Isabelle is able to infer the remaining set of facts given our
definitions above in order to check the unmodified proofs from the original
encoding.

\section{Conclusion}

In this Chapter, we motivated our rationale behind relaxing the network model
on top of which we verify our \CRDTs. We described our proof strategy for
relaxing the network model, and presented two example \CRDTs which we verified
on top of this network model. We began each example by showing the state-based
object, and a proof that each \CRDT inhabits the \SEC locale, even on the
relaxed network model.

We then reasoned that our state-based example \CRDTs in fact establish the same
goal for $\delta$-state \CRDTs without additional modification, according to our
result in Maxim~\ref{maxim:delta-as-op}. Finally, we presented an alternate
encoding which restricts the set of messages nodes are allowed to send together,
which more closely approximates the set of messages that $\delta$-\CRDTs are
allowed to send, and again proved that this encoding is an inhabitant of \SEC.

  \chapter{Future Work}
\label{chap:future-work}

This chapter outlines potential future research directions based on interesting
and under-explored areas in this work. Here, we will outline three directions in
the area of verifying $\delta$-state \CRDTs, as well as some insight that might
be gained by exploring each of these directions. It is our hope that future
researchers in this area may choose to conduct further investigation into these
areas.

\section{Verifying additional $\delta$-state \CRDTs}
\label{sec:future-pair-locale}
In our work, we presented examples of two $\delta$-state \CRDTs: the $\delta$
G-Counter, and the $\delta$ G-Set. An immediate future direction is to
investigate and verify more instances of $\delta$-state CRDTs.

One area of particular interest is in the \emph{composition} of multiple
$\delta$-state CRDTs. We have begun investigating the instantiation of a $pair$
locale, which takes as arguments two independent $\delta$-state \CRDTs, known as
``left'' and ``right.'' Our hope is that provided existing instantiations of
both of the sub-\CRDTs, that a $pair$ locale given two already-verified \CRDTs
could be used without additional proof burden to create another instance of the
$network{\isacharunderscore}with{\isacharunderscore}ops$ locale. That is: can
two already-verified $\delta$-state \CRDTs be used to compose a new
$\delta$-state CRDT which is their product without additional proof burden?

If this were possible, two new \CRDTs would be verified without effort: the
PN-Counter and the 2P-Set. These two \CRDTs are the most straightforward
composition of other known \CRDTs. Namely, the PN-Counter supports both an
\textsf{inc} and \textsf{dec} operation by maintaining \emph{two} counters (each
of which is treated as a single G-Counter, so the overall state is still
monotone and thus forms a join semi-lattice).

The PN-Counter has two $\delta$-state based G-Counter, which we refer to as
$(\fst s)$ and $(\snd s)$, where $s \in S$ refers to the state of the
PN-Counter. One possible specification for a $\delta$-based PN-Counter follows:

\begin{figure}[H]
  \[
    \textsf{PN-Counter}_\delta = \left\{\begin{aligned}
      S &: \mathbb{N}_0^{|\mathcal{I}|} \times \mathbb{N}_+^{|\mathcal{I}|} \\
      s^0 &: \left[ 0, 0, \cdots, 0 \right] \times \left[ 0, 0, \cdots, 0 \right] \\
      q^\delta &: \lambda s. \sum_{i \in \mathcal{I}} (\fst~s)(i) - \sum_{i \in
        \mathcal{I}} (\snd~s)(i) \\
      u^\delta &: \lambda s,(i,op). \begin{cases}
             \{ i \mapsto 1 \} \times \emptyset & \text{if $op=+$} \\
             \emptyset \times \{ i \mapsto 1 \} & \text{if $op=-$} \\
           \end{cases} \\
      m^\delta &: \lambda s_1, s_2.\, \begin{aligned}
             &\left\{ \max\{ i_1, i_2 \} : i \in \mathsf{dom}((\fst~s_1) \cup (\fst~s_1)) \right\} \\
             \times &\left\{ \max\{ i_1, i_2 \} : i \in \mathsf{dom}((\snd~s_1) \cup (\snd~s_1)) \right\}
           \end{aligned}
    \end{aligned}\right.
  \]
  \caption{$\delta$-state based \textsf{PN-Counter} \CRDT}
\end{figure}

Minor additional consideration is given to the updating function,
$u^\delta$, which returns an empty-set on the counter \emph{not} being updated.
Finally, the merging function $m^\delta$ merges the left- and right-hand sides
of the counter separately, and returns a pair. The 2P-Set is similar in function
to the above, substituting a $\delta$-state based G-Set in place of the
G-Counter.

If such a $\isa{pair}$ locale exists, we believe it would be as straightforward
as instantiating this locale over two copies of the G-Counter and G-Set to
obtain the PN-Counter and 2P-Set immediately.

\section{Direct $\delta$-state \CRDT proofs}
\label{sec:direct-delta-proofs}

To explore this idea, we drew significant inspiration from the work of Almedia
and his co-authors in~\citet{almedia18} to restate $\delta$-state \CRDTs in
terms of op-based \CRDTs in an effort to reuse as much of their library as
possible.

A significant drawback of this approach is that we are bound to the same
restrictions as op-based \CRDTs, which are inherently more restricted than
state-based \CRDTs. Much of this restriction comes from the \emph{eventual
delivery} property of \EC, which states that~\citep{shapiro11}:
\[
  \forall i, j.\, f \in c_i \Rightarrow \lozenge f \in c_j
\]
or that for any pair of correct replicas $i$, $j$ with histories $c_i$ and
$c_j$, respectively, an update received at one of those replicas is eventually
received at all other replicas.

Of course, under relaxed delivery semantics (i.e., in the case that the network
may delay messages for an infinite amount of time), op-based \CRDTs do not
achieve this property~\citep{shapiro11}. Namely, if an operation is performed at
some replica, and that message is dropped while in transit to another replica,
that replica will never receive the message.

State-based \CRDTs do not suffer from this problem, since \emph{every} update
they send encapsulates the history of all previous updates, since each update is
either reflected in the state, or subsumed by some later update which is itself
reflected in the state~\citep{shapiro11}. Since the entirety of the state is
shared with each replica during an update, state-based \CRDTs do not need to
impose an additional delivery relation in order to prove that they achieve \SEC.

op-based \CRDTs, on the other hand, \emph{do} need to specify an additional
delivery relation on top of their definition. That is, the delivery relation $P$
is a \emph{predicate} over network behaviors in which the eventual delivery
property can hold. In other words, for op-based \CRDTs:
\[
  P \Rightarrow \forall i, j.\, f \in c_i \Rightarrow \lozenge f \in c_j
\]
where it is a standard assumption that $P$ preserves (1) message order up to
concurrent messages and (2) at least once delivery~\citep{shapiro11,
almedia18}.\footnote{In practice, vector timestamps or globally unique
identifiers are associated with each message at the network layer, and messages
are reordered upon delivery to ensure that messages are delivered in the correct
order. Since all messages are eventually delivered under the precondition $P$,
this is a standard assumption.}

However, specifying $\delta$-based \CRDTs as a refinement of state-based \CRDTs
directly would not be sufficient for a constructive proof that $\delta$-state
based \CRDTs achieve \SEC. This is due to the fact that $\delta$-state \CRDTs
send state \emph{fragments}, which makes them the state-based analogue of
op-based \CRDTs. Without an additional delivery relation, $\delta$-state \CRDT
replicas which do not receive some update will never catch up without additional
updates.

Consider the figure below. In this example, we have three $\delta$-state \CRDT
replicas of a $\delta$-based GCounter, and an \textsf{inc} is performed at
$r_1$. Immediately, $r_1$ generates the state fragment $\{ r_1 \mapsto 1 \}$,
and sends it to the other replicas, $r_2$ and $r_3$. For the sake of example,
say that the network drops the update in route to $r_3$ such that it is never
received by $r_3$:

\begin{figure}[H]
  \centering
  \includegraphics[width=.6\textwidth]{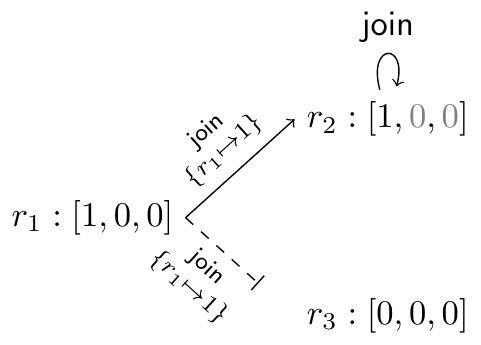}
  \caption{$\delta$-state \CRDTs violating \SEC without the causal merging
    condition.}
  \label{fig:delta-sec-violation}
\end{figure}

Without any future updates, neither of the replicas that \emph{have} received
the update will ever have reason to update $r_3$ again. This is a demonstration
of a \SEC violation, since:
\[
  \{ r_1 \mapsto 1 \} \in r_i \nRightarrow \lozenge \{ r_1 \mapsto 1 \} \in r_3,
  \quad i \in \{ 1, 2 \}
\]
That is, though the update $\{ r_1 \mapsto 1 \}$ is in the node histories of
$r_1$ and $r_2$ (both of which are behaving correctly), that update is
\emph{never} in the history of node $r_3$, which is also behaving correctly.

A critical issue in the above example is that $r_2$ merges the update from $r_1$
immediately--thus placing the update in that node's history--without knowing
whether or not it has been received by $r_3$. An anti-entropy algorithm like
in~\citet{almedia18} addresses these problems. The goal of an anti-entropy
algorithm for $\delta$-state \CRDTs is to do the following:
\begin{enumerate}
  \item On an operation, generate the $\delta$-mutation, and apply it to both
    the local state, and a temporary \emph{$\delta$-group}.
  \item Periodically, randomly choose between the current state and current
    $\delta$-group, and send its entire contents to all other replicas, and
    flush the $\delta$-group.
\end{enumerate}
This ensures that--even without outside interaction--the system as in
Figure~\ref{fig:delta-sec-violation} will eventually recover. This follows since
either one of $r_1$ or $r_2$ will at some point send their full state to all
other replicas, including $r_3$, at which point $r_3$ will have caught up.

We believe that it would be a worthwhile research goal to encode this
anti-entropy algorithm into a proof assistant, and specify that $\delta$-state
\CRDTs achieve \SEC \emph{without} a correspondence to traditional op-based
\CRDTs. Similarly to our work, in which we found a correspondence between
$\delta$- and op-based \CRDTs, we believe that specifying $\delta$-state \CRDTs
on their own would highlight the ways in which $\delta$-state \CRDTs are
\emph{different} from op-based \CRDTs.

Likewise, specifying the goal in this fashion would allow the proof to reason
about more network behaviors without a delivery predicate $P$, since the proof
would be aided by the periodic behavior of the anti-entropy algorithm above.

\section{Causally Consistent $\delta$-\CRDTs}

Another difference between op- and state-based \CRDTs is that state-based \CRDTs
require a \emph{causal merging condition} in order to ensure causal consistency
(that is, that updates are applied in a fashion that preserves their causality),
whereas in op-based \CRDTs this is traditionally an assumption placed on
$P$~\citep{shapiro11}.

The authors of~\citet{almedia18} define a \emph{$\delta$-interval}
$\Delta^{a,b}_i$ as:
\[
  \Delta^{a,b}_i = \bigsqcup \left\{ d_i^k : k \in [a, b) \right\}
\]
that is, $\Delta^{a,b}_i$ contains the deltas that occurred at replica $i$
beginning at time $a$ and up until time $b$. They go on to use this
$\delta$-interval to define the \emph{causal merging condition}, which is that
replica $i$ only joins a $\delta$-interval $\Delta^{a,b}_j$ into its own state
$X_i$ if:
\[
  X_i \sqsupseteq X_j^a
\]
That is, updates are only applied locally if they occurred \emph{before} the
latest-known update at replica $i$.

Algorithms which uphold the causal merging condition on $\delta$-intervals have
been proven on paper to satisfy Causal Consistency (\CC) in addition to \SEC. To
our knowledge, this result has not been mechanized, and so we believe it would
be a worthwhile direction of future research to specify the causal merging
condition and associated anti-entropy algorithms which preserve it into an
interactive theorem prover and mechanize the results of~\citet{almedia18}.

If the above is the subject of further exploration, we believe that it would be
additionally possible to prove that $\delta$-state \CRDTs achieve \SEC by a
\emph{simulation proof} that establishes their correspondence with state-based
\CRDTs. This is mentioned as Proposition 3 in~\citet{almedia18}, but we believe
that this is another fruitful area for formal verification.

  \chapter{Conclusion}
\label{chap:conclusion}

In this thesis, we extended the work in~\citet{gomes17} to provide a mechanized
proofs that $\delta$-\CRDTs~\citep{almedia18} achieve \SEC~\citep{shapiro11}.

Our central intuition (cf., Sections~\ref{sec:state-as-op}
and~\ref{sec:delta-as-op}) was to treat $\delta$- and state-based \CRDTs as
refinements of op-based \CRDTs.  This allowed us to successfully verify that two
\CRDTs--the G-Counter, and G-Set--achieve \SEC when specified both in the state-
and $\delta$-state based style.

In addition, we relaxed the network model by removing an assumption that all
messages are unique. While our main result is still predicated on a set of nice
delivery semantics $P$, this allowed us to quantify over an expanded set of all
possible network executions.

Together, this allowed us to restate the main result of~\citet{almedia18} in a
mechanized fashion. We believe that $\delta$-state \CRDTs satisfy an appealing
``best-of-both-worlds'' property. $\delta$-state \CRDTs require relatively
little of the network (like op-based \CRDTs), yet still maintain a relatively
small payload size (like state-based \CRDTs). This places great interest on
formal verification of their convergence properties.

In the future, we hope to see our result extended by specifying $\delta$-state
\CRDTs in terms of their state-based counterparts, as well as mechanizing
well-known anti-entropy algorithms and causality constraints on applying updates
from other replicas~\citep{almedia18}. We believe that this would be sufficient
to remove the precondition on a set of delivery semantics $P$ from our result.

  \appendix
  \chapter{Additional Proofs}

In this appendix, we provide the full proof scripts used in this work. The
source is available for free at: \url{https://github.com/ttaylorr/thesis}.

\section{state-based G-Counter \CRDT}
\label{sec:app-gcounter-comm-assoc}
\begin{isabellebody}
\isacommand{locale}\isamarkupfalse%
\ gcounter\ {\isacharequal}\ network{\isacharunderscore}with{\isacharunderscore}ops\ {\isacharunderscore}\ gcounter{\isacharunderscore}op\ {\isachardoublequoteopen}{\isasymlambda}\ x{\isachardot}\ None{\isachardoublequoteclose}\isanewline
\isanewline
\isacommand{lemma}\isamarkupfalse%
\ {\isacharparenleft}\isakeyword{in}\ gcounter{\isacharparenright}\ option{\isacharunderscore}max{\isacharunderscore}assoc{\isacharcolon}\isanewline
\ \ {\isachardoublequoteopen}option{\isacharunderscore}max\ a\ {\isacharparenleft}option{\isacharunderscore}max\ b\ c{\isacharparenright}\ {\isacharequal}\ option{\isacharunderscore}max\ {\isacharparenleft}option{\isacharunderscore}max\ a\ b{\isacharparenright}\ c{\isachardoublequoteclose}\isanewline
\isadelimproof
\ \ %
\endisadelimproof
\isatagproof
\isacommand{apply}\isamarkupfalse%
\ {\isacharparenleft}induction\ a{\isacharsemicolon}\ induction\ b{\isacharsemicolon}\ induction\ c{\isacharparenright}\isanewline
\ \ \isacommand{apply}\isamarkupfalse%
\ {\isacharparenleft}auto{\isacharparenright}\isanewline
\ \ \isacommand{done}\isamarkupfalse%
\endisatagproof
{\isafoldproof}%
\isadelimproof
\isanewline
\endisadelimproof
\isanewline
\isacommand{lemma}\isamarkupfalse%
\ {\isacharparenleft}\isakeyword{in}\ gcounter{\isacharparenright}\ option{\isacharunderscore}max{\isacharunderscore}commut{\isacharcolon}\ {\isachardoublequoteopen}option{\isacharunderscore}max\ a\ b\ {\isacharequal}\ option{\isacharunderscore}max\ b\ a{\isachardoublequoteclose}\isanewline
\isadelimproof
\ \ %
\endisadelimproof
\isatagproof
\isacommand{apply}\isamarkupfalse%
\ {\isacharparenleft}induction\ a{\isacharsemicolon}\ induction\ b{\isacharparenright}\isanewline
\ \ \isacommand{apply}\isamarkupfalse%
\ {\isacharparenleft}auto{\isacharparenright}\isanewline
\ \ \isacommand{done}\isamarkupfalse%
\endisatagproof
{\isafoldproof}%
\isadelimproof
\isanewline
\endisadelimproof
\isanewline
\isacommand{lemma}\isamarkupfalse%
\ {\isacharparenleft}\isakeyword{in}\ gcounter{\isacharparenright}\ {\isacharbrackleft}simp{\isacharbrackright}\ {\isacharcolon}\ {\isachardoublequoteopen}gcounter{\isacharunderscore}op\ x\ {\isasymrhd}\ gcounter{\isacharunderscore}op\ y\ {\isacharequal}\ gcounter{\isacharunderscore}op\ y\ {\isasymrhd}\ gcounter{\isacharunderscore}op\ x{\isachardoublequoteclose}\isanewline
\isadelimproof
\ \ %
\endisadelimproof
\isatagproof
\isacommand{apply}\isamarkupfalse%
\ {\isacharparenleft}auto\ simp\ add{\isacharcolon}\ kleisli{\isacharunderscore}def\ option{\isacharunderscore}max{\isacharunderscore}assoc{\isacharparenright}\isanewline
\ \ \isacommand{apply}\isamarkupfalse%
\ {\isacharparenleft}simp\ add{\isacharcolon}\ option{\isacharunderscore}max{\isacharunderscore}commut{\isacharparenright}\isanewline
\ \ \isacommand{done}\isamarkupfalse%
\endisatagproof
{\isafoldproof}%
\isadelimproof
\isanewline
\endisadelimproof
\isanewline
\isacommand{lemma}\isamarkupfalse%
\ {\isacharparenleft}\isakeyword{in}\ gcounter{\isacharparenright}\ concurrent{\isacharunderscore}operations{\isacharunderscore}commute{\isacharcolon}\isanewline
\ \ \isakeyword{assumes}\ {\isachardoublequoteopen}xs\ prefix\ of\ i{\isachardoublequoteclose}\isanewline
\ \ \isakeyword{shows}\ {\isachardoublequoteopen}hb{\isachardot}concurrent{\isacharunderscore}ops{\isacharunderscore}commute\ {\isacharparenleft}node{\isacharunderscore}deliver{\isacharunderscore}messages\ xs{\isacharparenright}{\isachardoublequoteclose}\isanewline
\isadelimproof
\ \ %
\endisadelimproof
\isatagproof
\isacommand{using}\isamarkupfalse%
\ assms\isanewline
\ \ \isacommand{apply}\isamarkupfalse%
{\isacharparenleft}clarsimp\ simp{\isacharcolon}\ hb{\isachardot}concurrent{\isacharunderscore}ops{\isacharunderscore}commute{\isacharunderscore}def{\isacharparenright}\isanewline
\ \ \isacommand{apply}\isamarkupfalse%
{\isacharparenleft}unfold\ interp{\isacharunderscore}msg{\isacharunderscore}def{\isacharcomma}\ simp{\isacharparenright}\isanewline
\ \ \isacommand{done}\isamarkupfalse%
\endisatagproof
{\isafoldproof}%
\isadelimproof
\isanewline
\endisadelimproof
\isanewline
\isacommand{corollary}\isamarkupfalse%
\ {\isacharparenleft}\isakeyword{in}\ gcounter{\isacharparenright}\ counter{\isacharunderscore}convergence{\isacharcolon}\isanewline
\ \ \isakeyword{assumes}\ {\isachardoublequoteopen}set\ {\isacharparenleft}node{\isacharunderscore}deliver{\isacharunderscore}messages\ xs{\isacharparenright}\ {\isacharequal}\ set\ {\isacharparenleft}node{\isacharunderscore}deliver{\isacharunderscore}messages\ ys{\isacharparenright}{\isachardoublequoteclose}\isanewline
\ \ \ \ \ \ \isakeyword{and}\ {\isachardoublequoteopen}xs\ prefix\ of\ i{\isachardoublequoteclose}\isanewline
\ \ \ \ \ \ \isakeyword{and}\ {\isachardoublequoteopen}ys\ prefix\ of\ j{\isachardoublequoteclose}\isanewline
\ \ \ \ \isakeyword{shows}\ {\isachardoublequoteopen}apply{\isacharunderscore}operations\ xs\ {\isacharequal}\ apply{\isacharunderscore}operations\ ys{\isachardoublequoteclose}\isanewline
\isadelimproof
\endisadelimproof
\isatagproof
\isacommand{using}\isamarkupfalse%
\ assms\ \isacommand{by}\isamarkupfalse%
{\isacharparenleft}auto\ simp\ add{\isacharcolon}\ apply{\isacharunderscore}operations{\isacharunderscore}def\isanewline
\ \ \ \ \ \ \ \ \ \ \ \ \ \ \ \ \ \ \ \ \ \ \ intro{\isacharcolon}\ hb{\isachardot}convergence{\isacharunderscore}ext\ concurrent{\isacharunderscore}operations{\isacharunderscore}commute\isanewline
\ \ \ \ \ \ \ \ \ \ \ \ \ \ \ \ \ \ \ \ \ \ \ \ \ \ \ \ \ \ node{\isacharunderscore}deliver{\isacharunderscore}messages{\isacharunderscore}distinct\ hb{\isacharunderscore}consistent{\isacharunderscore}prefix{\isacharparenright}%
\endisatagproof
{\isafoldproof}%
\isadelimproof
\isanewline
\endisadelimproof
\isanewline
\isacommand{context}\isamarkupfalse%
\ gcounter\ \isakeyword{begin}\isanewline
\isanewline
\isacommand{sublocale}\isamarkupfalse%
\ sec{\isacharcolon}\ strong{\isacharunderscore}eventual{\isacharunderscore}consistency\ weak{\isacharunderscore}hb\ hb\ interp{\isacharunderscore}msg\isanewline
\ \ {\isachardoublequoteopen}{\isasymlambda}ops{\isachardot}\ {\isasymexists}xs\ i{\isachardot}\ xs\ prefix\ of\ i\ {\isasymand}\ node{\isacharunderscore}deliver{\isacharunderscore}messages\ xs\ {\isacharequal}\ ops{\isachardoublequoteclose}\ {\isachardoublequoteopen}{\isasymlambda}\ x{\isachardot}\ None{\isachardoublequoteclose}\isanewline
\isadelimproof
\ \ %
\endisadelimproof
\isatagproof
\isacommand{apply}\isamarkupfalse%
{\isacharparenleft}standard{\isacharsemicolon}\ clarsimp{\isacharparenright}\isanewline
\ \ \ \ \ \ \isacommand{apply}\isamarkupfalse%
{\isacharparenleft}auto\ simp\ add{\isacharcolon}\ hb{\isacharunderscore}consistent{\isacharunderscore}prefix\ drop{\isacharunderscore}last{\isacharunderscore}message\ node{\isacharunderscore}deliver{\isacharunderscore}messages{\isacharunderscore}distinct\ concurrent{\isacharunderscore}operations{\isacharunderscore}commute{\isacharparenright}\isanewline
\ \ \ \isacommand{apply}\isamarkupfalse%
{\isacharparenleft}metis\ {\isacharparenleft}full{\isacharunderscore}types{\isacharparenright}\ interp{\isacharunderscore}msg{\isacharunderscore}def\ gcounter{\isacharunderscore}op{\isachardot}elims{\isacharparenright}\isanewline
\ \ \isacommand{using}\isamarkupfalse%
\ drop{\isacharunderscore}last{\isacharunderscore}message\ \isacommand{apply}\isamarkupfalse%
\ blast\isanewline
\ \ \isacommand{done}\isamarkupfalse%
\endisatagproof
{\isafoldproof}%
\isadelimproof
\isanewline
\endisadelimproof
\isacommand{end}\isamarkupfalse%
\isanewline
\isadelimtheory
\isanewline
\endisadelimtheory
\isatagtheory
\isacommand{end}\isamarkupfalse%
\endisatagtheory
{\isafoldtheory}%
\isadelimtheory
\endisadelimtheory
\end{isabellebody}

\section{state-based G-Set \CRDT}
\label{sec:app-state-gset}
\begin{isabellebody}
\isacommand{locale}\isamarkupfalse%
\ gset\ {\isacharequal}\ network{\isacharunderscore}with{\isacharunderscore}ops\ {\isacharunderscore}\ gset{\isacharunderscore}op\ {\isachardoublequoteopen}{\isacharbraceleft}{\isacharbraceright}{\isachardoublequoteclose}\isanewline
\isanewline
\isacommand{lemma}\isamarkupfalse%
\ {\isacharparenleft}\isakeyword{in}\ gset{\isacharparenright}\ {\isacharbrackleft}simp{\isacharbrackright}\ {\isacharcolon}\ {\isachardoublequoteopen}gset{\isacharunderscore}op\ x\ {\isasymrhd}\ gset{\isacharunderscore}op\ y\ {\isacharequal}\ gset{\isacharunderscore}op\ y\ {\isasymrhd}\ gset{\isacharunderscore}op\ x{\isachardoublequoteclose}\isanewline
\isadelimproof
\ \ %
\endisadelimproof
\isatagproof
\isacommand{apply}\isamarkupfalse%
\ {\isacharparenleft}auto\ simp\ add{\isacharcolon}\ kleisli{\isacharunderscore}def{\isacharparenright}\isanewline
\isacommand{done}\isamarkupfalse%
\endisatagproof
{\isafoldproof}%
\isadelimproof
\isanewline
\endisadelimproof
\isanewline
\isacommand{lemma}\isamarkupfalse%
\ {\isacharparenleft}\isakeyword{in}\ gset{\isacharparenright}\ concurrent{\isacharunderscore}operations{\isacharunderscore}commute{\isacharcolon}\isanewline
\ \ \isakeyword{assumes}\ {\isachardoublequoteopen}xs\ prefix\ of\ i{\isachardoublequoteclose}\isanewline
\ \ \isakeyword{shows}\ {\isachardoublequoteopen}hb{\isachardot}concurrent{\isacharunderscore}ops{\isacharunderscore}commute\ {\isacharparenleft}node{\isacharunderscore}deliver{\isacharunderscore}messages\ xs{\isacharparenright}{\isachardoublequoteclose}\isanewline
\isadelimproof
\ \ %
\endisadelimproof
\isatagproof
\isacommand{using}\isamarkupfalse%
\ assms\isanewline
\ \ \isacommand{apply}\isamarkupfalse%
{\isacharparenleft}clarsimp\ simp{\isacharcolon}\ hb{\isachardot}concurrent{\isacharunderscore}ops{\isacharunderscore}commute{\isacharunderscore}def{\isacharparenright}\isanewline
\ \ \isacommand{apply}\isamarkupfalse%
{\isacharparenleft}unfold\ interp{\isacharunderscore}msg{\isacharunderscore}def{\isacharcomma}\ simp{\isacharparenright}\isanewline
\isacommand{done}\isamarkupfalse%
\endisatagproof
{\isafoldproof}%
\isadelimproof
\isanewline
\endisadelimproof
\isanewline
\isacommand{corollary}\isamarkupfalse%
\ {\isacharparenleft}\isakeyword{in}\ gset{\isacharparenright}\ set{\isacharunderscore}convergence{\isacharcolon}\isanewline
\ \ \isakeyword{assumes}\ {\isachardoublequoteopen}set\ {\isacharparenleft}node{\isacharunderscore}deliver{\isacharunderscore}messages\ xs{\isacharparenright}\ {\isacharequal}\ set\ {\isacharparenleft}node{\isacharunderscore}deliver{\isacharunderscore}messages\ ys{\isacharparenright}{\isachardoublequoteclose}\isanewline
\ \ \ \ \ \ \isakeyword{and}\ {\isachardoublequoteopen}xs\ prefix\ of\ i{\isachardoublequoteclose}\isanewline
\ \ \ \ \ \ \isakeyword{and}\ {\isachardoublequoteopen}ys\ prefix\ of\ j{\isachardoublequoteclose}\isanewline
\ \ \ \ \isakeyword{shows}\ {\isachardoublequoteopen}apply{\isacharunderscore}operations\ xs\ {\isacharequal}\ apply{\isacharunderscore}operations\ ys{\isachardoublequoteclose}\isanewline
\isadelimproof
\endisadelimproof
\isatagproof
\isacommand{using}\isamarkupfalse%
\ assms\ \isacommand{by}\isamarkupfalse%
{\isacharparenleft}auto\ simp\ add{\isacharcolon}\ apply{\isacharunderscore}operations{\isacharunderscore}def\ intro{\isacharcolon}\ hb{\isachardot}convergence{\isacharunderscore}ext\ concurrent{\isacharunderscore}operations{\isacharunderscore}commute\isanewline
\ \ \ \ \ \ \ \ \ \ \ \ \ \ \ \ node{\isacharunderscore}deliver{\isacharunderscore}messages{\isacharunderscore}distinct\ hb{\isacharunderscore}consistent{\isacharunderscore}prefix{\isacharparenright}%
\endisatagproof
{\isafoldproof}%
\isadelimproof
\isanewline
\endisadelimproof
\isanewline
\isacommand{context}\isamarkupfalse%
\ gset\ \isakeyword{begin}\isanewline
\isanewline
\isacommand{sublocale}\isamarkupfalse%
\ sec{\isacharcolon}\ strong{\isacharunderscore}eventual{\isacharunderscore}consistency\ weak{\isacharunderscore}hb\ hb\ interp{\isacharunderscore}msg\isanewline
\ \ {\isachardoublequoteopen}{\isasymlambda}ops{\isachardot}\ {\isasymexists}xs\ i{\isachardot}\ xs\ prefix\ of\ i\ {\isasymand}\ node{\isacharunderscore}deliver{\isacharunderscore}messages\ xs\ {\isacharequal}\ ops{\isachardoublequoteclose}\ {\isachardoublequoteopen}{\isacharbraceleft}{\isacharbraceright}{\isachardoublequoteclose}\isanewline
\isadelimproof
\ \ %
\endisadelimproof
\isatagproof
\isacommand{apply}\isamarkupfalse%
{\isacharparenleft}standard{\isacharsemicolon}\ clarsimp{\isacharparenright}\isanewline
\ \ \ \ \ \ \isacommand{apply}\isamarkupfalse%
{\isacharparenleft}auto\ simp\ add{\isacharcolon}\ hb{\isacharunderscore}consistent{\isacharunderscore}prefix\ drop{\isacharunderscore}last{\isacharunderscore}message\ node{\isacharunderscore}deliver{\isacharunderscore}messages{\isacharunderscore}distinct\ concurrent{\isacharunderscore}operations{\isacharunderscore}commute{\isacharparenright}\isanewline
\ \ \ \isacommand{apply}\isamarkupfalse%
{\isacharparenleft}metis\ {\isacharparenleft}full{\isacharunderscore}types{\isacharparenright}\ interp{\isacharunderscore}msg{\isacharunderscore}def\ gset{\isacharunderscore}op{\isachardot}elims{\isacharparenright}\isanewline
\ \ \isacommand{using}\isamarkupfalse%
\ drop{\isacharunderscore}last{\isacharunderscore}message\ \isacommand{apply}\isamarkupfalse%
\ blast\isanewline
\ \ \isacommand{done}\isamarkupfalse%
\endisatagproof
{\isafoldproof}%
\isadelimproof
\isanewline
\endisadelimproof
\isacommand{end}\isamarkupfalse%
\isanewline
\isadelimtheory
\isanewline
\endisadelimtheory
\isatagtheory
\isacommand{end}\isamarkupfalse%
\endisatagtheory
{\isafoldtheory}%
\isadelimtheory
\endisadelimtheory
\end{isabellebody}%

\section{$\delta$-state G-Counter \CRDT}
\begin{isabellebody}
\isacommand{locale}\isamarkupfalse%
\ delta{\isacharunderscore}gcounter\ {\isacharequal}\ network{\isacharunderscore}with{\isacharunderscore}ops\ {\isacharunderscore}\ delta{\isacharunderscore}gcounter{\isacharunderscore}op\ {\isachardoublequoteopen}{\isasymlambda}\ x{\isachardot}\ None{\isachardoublequoteclose}\isanewline
\isanewline
\isacommand{lemma}\isamarkupfalse%
\ {\isacharparenleft}\isakeyword{in}\ delta{\isacharunderscore}gcounter{\isacharparenright}\ option{\isacharunderscore}max{\isacharunderscore}assoc{\isacharcolon}\isanewline
\ \ {\isachardoublequoteopen}option{\isacharunderscore}max\ a\ {\isacharparenleft}option{\isacharunderscore}max\ b\ c{\isacharparenright}\ {\isacharequal}\ option{\isacharunderscore}max\ {\isacharparenleft}option{\isacharunderscore}max\ a\ b{\isacharparenright}\ c{\isachardoublequoteclose}\isanewline
\isadelimproof
\ \ %
\endisadelimproof
\isatagproof
\isacommand{apply}\isamarkupfalse%
\ {\isacharparenleft}induction\ a{\isacharsemicolon}\ induction\ b{\isacharsemicolon}\ induction\ c{\isacharparenright}\isanewline
\ \ \isacommand{apply}\isamarkupfalse%
\ {\isacharparenleft}auto{\isacharparenright}\isanewline
\ \ \isacommand{done}\isamarkupfalse%
\endisatagproof
{\isafoldproof}%
\isadelimproof
\isanewline
\endisadelimproof
\isanewline
\isacommand{lemma}\isamarkupfalse%
\ {\isacharparenleft}\isakeyword{in}\ delta{\isacharunderscore}gcounter{\isacharparenright}\ option{\isacharunderscore}max{\isacharunderscore}commut{\isacharcolon}\ {\isachardoublequoteopen}option{\isacharunderscore}max\ a\ b\ {\isacharequal}\ option{\isacharunderscore}max\ b\ a{\isachardoublequoteclose}\isanewline
\isadelimproof
\ \ %
\endisadelimproof
\isatagproof
\isacommand{apply}\isamarkupfalse%
\ {\isacharparenleft}induction\ a{\isacharsemicolon}\ induction\ b{\isacharparenright}\isanewline
\ \ \isacommand{apply}\isamarkupfalse%
\ {\isacharparenleft}auto{\isacharparenright}\isanewline
\ \ \isacommand{done}\isamarkupfalse%
\endisatagproof
{\isafoldproof}%
\isadelimproof
\isanewline
\endisadelimproof
\isanewline
\isacommand{lemma}\isamarkupfalse%
\ {\isacharparenleft}\isakeyword{in}\ delta{\isacharunderscore}gcounter{\isacharparenright}\ {\isacharbrackleft}simp{\isacharbrackright}\ {\isacharcolon}\ {\isachardoublequoteopen}delta{\isacharunderscore}gcounter{\isacharunderscore}op\ x\ {\isasymrhd}\ delta{\isacharunderscore}gcounter{\isacharunderscore}op\ y\ {\isacharequal}\ delta{\isacharunderscore}gcounter{\isacharunderscore}op\ y\ {\isasymrhd}\ delta{\isacharunderscore}gcounter{\isacharunderscore}op\ x{\isachardoublequoteclose}\isanewline
\isadelimproof
\ \ %
\endisadelimproof
\isatagproof
\isacommand{apply}\isamarkupfalse%
\ {\isacharparenleft}auto\ simp\ add{\isacharcolon}\ kleisli{\isacharunderscore}def\ option{\isacharunderscore}max{\isacharunderscore}assoc{\isacharparenright}\isanewline
\ \ \isacommand{apply}\isamarkupfalse%
\ {\isacharparenleft}simp\ add{\isacharcolon}\ option{\isacharunderscore}max{\isacharunderscore}commut{\isacharparenright}\isanewline
\ \ \isacommand{done}\isamarkupfalse%
\endisatagproof
{\isafoldproof}%
\isadelimproof
\isanewline
\endisadelimproof
\isanewline
\isacommand{lemma}\isamarkupfalse%
\ {\isacharparenleft}\isakeyword{in}\ delta{\isacharunderscore}gcounter{\isacharparenright}\ concurrent{\isacharunderscore}operations{\isacharunderscore}commute{\isacharcolon}\isanewline
\ \ \isakeyword{assumes}\ {\isachardoublequoteopen}xs\ prefix\ of\ i{\isachardoublequoteclose}\isanewline
\ \ \isakeyword{shows}\ {\isachardoublequoteopen}hb{\isachardot}concurrent{\isacharunderscore}ops{\isacharunderscore}commute\ {\isacharparenleft}node{\isacharunderscore}deliver{\isacharunderscore}messages\ xs{\isacharparenright}{\isachardoublequoteclose}\isanewline
\isadelimproof
\ \ %
\endisadelimproof
\isatagproof
\isacommand{using}\isamarkupfalse%
\ assms\isanewline
\ \ \isacommand{apply}\isamarkupfalse%
{\isacharparenleft}clarsimp\ simp{\isacharcolon}\ hb{\isachardot}concurrent{\isacharunderscore}ops{\isacharunderscore}commute{\isacharunderscore}def{\isacharparenright}\isanewline
\ \ \isacommand{apply}\isamarkupfalse%
{\isacharparenleft}unfold\ interp{\isacharunderscore}msg{\isacharunderscore}def{\isacharcomma}\ simp{\isacharparenright}\isanewline
\ \ \isacommand{done}\isamarkupfalse%
\endisatagproof
{\isafoldproof}%
\isadelimproof
\isanewline
\endisadelimproof
\isanewline
\isacommand{corollary}\isamarkupfalse%
\ {\isacharparenleft}\isakeyword{in}\ delta{\isacharunderscore}gcounter{\isacharparenright}\ counter{\isacharunderscore}convergence{\isacharcolon}\isanewline
\ \ \isakeyword{assumes}\ {\isachardoublequoteopen}set\ {\isacharparenleft}node{\isacharunderscore}deliver{\isacharunderscore}messages\ xs{\isacharparenright}\ {\isacharequal}\ set\ {\isacharparenleft}node{\isacharunderscore}deliver{\isacharunderscore}messages\ ys{\isacharparenright}{\isachardoublequoteclose}\isanewline
\ \ \ \ \ \ \isakeyword{and}\ {\isachardoublequoteopen}xs\ prefix\ of\ i{\isachardoublequoteclose}\isanewline
\ \ \ \ \ \ \isakeyword{and}\ {\isachardoublequoteopen}ys\ prefix\ of\ j{\isachardoublequoteclose}\isanewline
\ \ \ \ \isakeyword{shows}\ {\isachardoublequoteopen}apply{\isacharunderscore}operations\ xs\ {\isacharequal}\ apply{\isacharunderscore}operations\ ys{\isachardoublequoteclose}\isanewline
\isadelimproof
\endisadelimproof
\isatagproof
\isacommand{using}\isamarkupfalse%
\ assms\ \isacommand{by}\isamarkupfalse%
{\isacharparenleft}auto\ simp\ add{\isacharcolon}\ apply{\isacharunderscore}operations{\isacharunderscore}def\ intro{\isacharcolon}\ hb{\isachardot}convergence{\isacharunderscore}ext\ concurrent{\isacharunderscore}operations{\isacharunderscore}commute\isanewline
\ \ \ \ \ \ \ \ \ \ \ \ \ \ \ \ node{\isacharunderscore}deliver{\isacharunderscore}messages{\isacharunderscore}distinct\ hb{\isacharunderscore}consistent{\isacharunderscore}prefix{\isacharparenright}%
\endisatagproof
{\isafoldproof}%
\isadelimproof
\isanewline
\endisadelimproof
\isanewline
\isacommand{context}\isamarkupfalse%
\ delta{\isacharunderscore}gcounter\ \isakeyword{begin}\isanewline
\isanewline
\isacommand{sublocale}\isamarkupfalse%
\ sec{\isacharcolon}\ strong{\isacharunderscore}eventual{\isacharunderscore}consistency\ weak{\isacharunderscore}hb\ hb\ interp{\isacharunderscore}msg\isanewline
\ \ {\isachardoublequoteopen}{\isasymlambda}ops{\isachardot}\ {\isasymexists}xs\ i{\isachardot}\ xs\ prefix\ of\ i\ {\isasymand}\ node{\isacharunderscore}deliver{\isacharunderscore}messages\ xs\ {\isacharequal}\ ops{\isachardoublequoteclose}\ {\isachardoublequoteopen}{\isasymlambda}\ x{\isachardot}\ None{\isachardoublequoteclose}\isanewline
\isadelimproof
\ \ %
\endisadelimproof
\isatagproof
\isacommand{apply}\isamarkupfalse%
{\isacharparenleft}standard{\isacharsemicolon}\ clarsimp{\isacharparenright}\isanewline
\ \ \ \ \ \ \isacommand{apply}\isamarkupfalse%
{\isacharparenleft}auto\ simp\ add{\isacharcolon}\ hb{\isacharunderscore}consistent{\isacharunderscore}prefix\ drop{\isacharunderscore}last{\isacharunderscore}message\ node{\isacharunderscore}deliver{\isacharunderscore}messages{\isacharunderscore}distinct\ concurrent{\isacharunderscore}operations{\isacharunderscore}commute{\isacharparenright}\isanewline
\ \ \ \isacommand{apply}\isamarkupfalse%
{\isacharparenleft}metis\ {\isacharparenleft}full{\isacharunderscore}types{\isacharparenright}\ interp{\isacharunderscore}msg{\isacharunderscore}def\ delta{\isacharunderscore}gcounter{\isacharunderscore}op{\isachardot}elims{\isacharparenright}\isanewline
\ \ \isacommand{using}\isamarkupfalse%
\ drop{\isacharunderscore}last{\isacharunderscore}message\ \isacommand{apply}\isamarkupfalse%
\ blast\isanewline
\ \ \isacommand{done}\isamarkupfalse%
\endisatagproof
{\isafoldproof}%
\isadelimproof
\isanewline
\endisadelimproof
\isacommand{end}\isamarkupfalse%
\isanewline
\isadelimtheory
\isanewline
\endisadelimtheory
\isatagtheory
\isacommand{end}\isamarkupfalse%
\endisatagtheory
{\isafoldtheory}%
\isadelimtheory
\endisadelimtheory
\end{isabellebody}%

\section{$\delta$-state G-Set \CRDT}
\label{sec:app-delta-gset}
\begin{isabellebody}%
\isacommand{locale}\isamarkupfalse%
\ delta{\isacharunderscore}gset\ {\isacharequal}\ network{\isacharunderscore}with{\isacharunderscore}ops\ {\isacharunderscore}\ delta{\isacharunderscore}gset{\isacharunderscore}op\ {\isachardoublequoteopen}{\isacharbraceleft}{\isacharbraceright}{\isachardoublequoteclose}\isanewline
\isanewline
\isacommand{lemma}\isamarkupfalse%
\ {\isacharparenleft}\isakeyword{in}\ delta{\isacharunderscore}gset{\isacharparenright}\ {\isacharbrackleft}simp{\isacharbrackright}\ {\isacharcolon}\ {\isachardoublequoteopen}delta{\isacharunderscore}gset{\isacharunderscore}op\ x\ {\isasymrhd}\ delta{\isacharunderscore}gset{\isacharunderscore}op\ y\ {\isacharequal}\ delta{\isacharunderscore}gset{\isacharunderscore}op\ y\ {\isasymrhd}\ delta{\isacharunderscore}gset{\isacharunderscore}op\ x{\isachardoublequoteclose}\isanewline
\isadelimproof
\ \ %
\endisadelimproof
\isatagproof
\isacommand{apply}\isamarkupfalse%
\ {\isacharparenleft}auto\ simp\ add{\isacharcolon}\ kleisli{\isacharunderscore}def{\isacharparenright}\isanewline
\isacommand{done}\isamarkupfalse%
\endisatagproof
{\isafoldproof}%
\isadelimproof
\isanewline
\endisadelimproof
\isanewline
\isacommand{lemma}\isamarkupfalse%
\ {\isacharparenleft}\isakeyword{in}\ delta{\isacharunderscore}gset{\isacharparenright}\ concurrent{\isacharunderscore}operations{\isacharunderscore}commute{\isacharcolon}\isanewline
\ \ \isakeyword{assumes}\ {\isachardoublequoteopen}xs\ prefix\ of\ i{\isachardoublequoteclose}\isanewline
\ \ \isakeyword{shows}\ {\isachardoublequoteopen}hb{\isachardot}concurrent{\isacharunderscore}ops{\isacharunderscore}commute\ {\isacharparenleft}node{\isacharunderscore}deliver{\isacharunderscore}messages\ xs{\isacharparenright}{\isachardoublequoteclose}\isanewline
\isadelimproof
\ \ %
\endisadelimproof
\isatagproof
\isacommand{using}\isamarkupfalse%
\ assms\isanewline
\ \ \isacommand{apply}\isamarkupfalse%
{\isacharparenleft}clarsimp\ simp{\isacharcolon}\ hb{\isachardot}concurrent{\isacharunderscore}ops{\isacharunderscore}commute{\isacharunderscore}def{\isacharparenright}\isanewline
\ \ \isacommand{apply}\isamarkupfalse%
{\isacharparenleft}unfold\ interp{\isacharunderscore}msg{\isacharunderscore}def{\isacharcomma}\ simp{\isacharparenright}\isanewline
\isacommand{done}\isamarkupfalse%
\endisatagproof
{\isafoldproof}%
\isadelimproof
\isanewline
\endisadelimproof
\isanewline
\isacommand{corollary}\isamarkupfalse%
\ {\isacharparenleft}\isakeyword{in}\ delta{\isacharunderscore}gset{\isacharparenright}\ set{\isacharunderscore}convergence{\isacharcolon}\isanewline
\ \ \isakeyword{assumes}\ {\isachardoublequoteopen}set\ {\isacharparenleft}node{\isacharunderscore}deliver{\isacharunderscore}messages\ xs{\isacharparenright}\ {\isacharequal}\ set\ {\isacharparenleft}node{\isacharunderscore}deliver{\isacharunderscore}messages\ ys{\isacharparenright}{\isachardoublequoteclose}\isanewline
\ \ \ \ \ \ \isakeyword{and}\ {\isachardoublequoteopen}xs\ prefix\ of\ i{\isachardoublequoteclose}\isanewline
\ \ \ \ \ \ \isakeyword{and}\ {\isachardoublequoteopen}ys\ prefix\ of\ j{\isachardoublequoteclose}\isanewline
\ \ \ \ \isakeyword{shows}\ {\isachardoublequoteopen}apply{\isacharunderscore}operations\ xs\ {\isacharequal}\ apply{\isacharunderscore}operations\ ys{\isachardoublequoteclose}\isanewline
\isadelimproof
\endisadelimproof
\isatagproof
\isacommand{using}\isamarkupfalse%
\ assms\ \isacommand{by}\isamarkupfalse%
{\isacharparenleft}auto\ simp\ add{\isacharcolon}\ apply{\isacharunderscore}operations{\isacharunderscore}def\ intro{\isacharcolon}\ hb{\isachardot}convergence{\isacharunderscore}ext\ concurrent{\isacharunderscore}operations{\isacharunderscore}commute\isanewline
\ \ \ \ \ \ \ \ \ \ \ \ \ \ \ \ node{\isacharunderscore}deliver{\isacharunderscore}messages{\isacharunderscore}distinct\ hb{\isacharunderscore}consistent{\isacharunderscore}prefix{\isacharparenright}%
\endisatagproof
{\isafoldproof}%
\isadelimproof
\isanewline
\endisadelimproof
\isanewline
\isacommand{context}\isamarkupfalse%
\ delta{\isacharunderscore}gset\ \isakeyword{begin}\isanewline
\isanewline
\isacommand{sublocale}\isamarkupfalse%
\ sec{\isacharcolon}\ strong{\isacharunderscore}eventual{\isacharunderscore}consistency\ weak{\isacharunderscore}hb\ hb\ interp{\isacharunderscore}msg\isanewline
\ \ {\isachardoublequoteopen}{\isasymlambda}ops{\isachardot}\ {\isasymexists}xs\ i{\isachardot}\ xs\ prefix\ of\ i\ {\isasymand}\ node{\isacharunderscore}deliver{\isacharunderscore}messages\ xs\ {\isacharequal}\ ops{\isachardoublequoteclose}\ {\isachardoublequoteopen}{\isacharbraceleft}{\isacharbraceright}{\isachardoublequoteclose}\isanewline
\isadelimproof
\ \ %
\endisadelimproof
\isatagproof
\isacommand{apply}\isamarkupfalse%
{\isacharparenleft}standard{\isacharsemicolon}\ clarsimp{\isacharparenright}\isanewline
\ \ \ \ \ \ \isacommand{apply}\isamarkupfalse%
{\isacharparenleft}auto\ simp\ add{\isacharcolon}\ hb{\isacharunderscore}consistent{\isacharunderscore}prefix\ drop{\isacharunderscore}last{\isacharunderscore}message\ node{\isacharunderscore}deliver{\isacharunderscore}messages{\isacharunderscore}distinct\ concurrent{\isacharunderscore}operations{\isacharunderscore}commute{\isacharparenright}\isanewline
\ \ \ \isacommand{apply}\isamarkupfalse%
{\isacharparenleft}metis\ {\isacharparenleft}full{\isacharunderscore}types{\isacharparenright}\ interp{\isacharunderscore}msg{\isacharunderscore}def\ delta{\isacharunderscore}gset{\isacharunderscore}op{\isachardot}elims{\isacharparenright}\isanewline
\ \ \isacommand{using}\isamarkupfalse%
\ drop{\isacharunderscore}last{\isacharunderscore}message\ \isacommand{apply}\isamarkupfalse%
\ blast\isanewline
\ \ \isacommand{done}\isamarkupfalse%
\endisatagproof
{\isafoldproof}%
\isadelimproof
\isanewline
\endisadelimproof
\isacommand{end}\isamarkupfalse%
\isanewline
\isadelimtheory
\isanewline
\endisadelimtheory
\isatagtheory
\isacommand{end}\isamarkupfalse%
\endisatagtheory
{\isafoldtheory}%
\isadelimtheory
\endisadelimtheory
\end{isabellebody}%

\section{Restricted $\delta$-state G-Counter \CRDT}
\begin{isabellebody}
\isacommand{locale}\isamarkupfalse%
\ delta{\isacharunderscore}gcounter\ {\isacharequal}\ network{\isacharunderscore}with{\isacharunderscore}ops\ {\isacharunderscore}\ delta{\isacharunderscore}gcounter{\isacharunderscore}op\ {\isachardoublequoteopen}{\isasymlambda}\ x{\isachardot}\ None{\isachardoublequoteclose}\isanewline
\isanewline
\isacommand{lemma}\isamarkupfalse%
\ {\isacharparenleft}\isakeyword{in}\ delta{\isacharunderscore}gcounter{\isacharparenright}\ option{\isacharunderscore}max{\isacharunderscore}assoc{\isacharcolon}\isanewline
\ \ {\isachardoublequoteopen}option{\isacharunderscore}max\ a\ {\isacharparenleft}option{\isacharunderscore}max\ b\ c{\isacharparenright}\ {\isacharequal}\ option{\isacharunderscore}max\ {\isacharparenleft}option{\isacharunderscore}max\ a\ b{\isacharparenright}\ c{\isachardoublequoteclose}\isanewline
\isadelimproof
\ \ %
\endisadelimproof
\isatagproof
\isacommand{apply}\isamarkupfalse%
\ {\isacharparenleft}induction\ a{\isacharsemicolon}\ induction\ b{\isacharsemicolon}\ induction\ c{\isacharparenright}\isanewline
\ \ \isacommand{apply}\isamarkupfalse%
\ {\isacharparenleft}auto{\isacharparenright}\isanewline
\ \ \isacommand{done}\isamarkupfalse%
\endisatagproof
{\isafoldproof}%
\isadelimproof
\isanewline
\endisadelimproof
\isanewline
\isacommand{lemma}\isamarkupfalse%
\ {\isacharparenleft}\isakeyword{in}\ delta{\isacharunderscore}gcounter{\isacharparenright}\ option{\isacharunderscore}max{\isacharunderscore}commut{\isacharcolon}\ {\isachardoublequoteopen}option{\isacharunderscore}max\ a\ b\ {\isacharequal}\ option{\isacharunderscore}max\ b\ a{\isachardoublequoteclose}\isanewline
\isadelimproof
\ \ %
\endisadelimproof
\isatagproof
\isacommand{apply}\isamarkupfalse%
\ {\isacharparenleft}induction\ a{\isacharsemicolon}\ induction\ b{\isacharparenright}\isanewline
\ \ \isacommand{apply}\isamarkupfalse%
\ {\isacharparenleft}auto{\isacharparenright}\isanewline
\ \ \isacommand{done}\isamarkupfalse%
\endisatagproof
{\isafoldproof}%
\isadelimproof
\isanewline
\endisadelimproof
\isanewline
\isacommand{lemma}\isamarkupfalse%
\ {\isacharparenleft}\isakeyword{in}\ delta{\isacharunderscore}gcounter{\isacharparenright}\ {\isacharbrackleft}simp{\isacharbrackright}\ {\isacharcolon}\ {\isachardoublequoteopen}delta{\isacharunderscore}gcounter{\isacharunderscore}op\ x\ {\isasymrhd}\ delta{\isacharunderscore}gcounter{\isacharunderscore}op\ y\ {\isacharequal}\ delta{\isacharunderscore}gcounter{\isacharunderscore}op\ y\ {\isasymrhd}\ delta{\isacharunderscore}gcounter{\isacharunderscore}op\ x{\isachardoublequoteclose}\isanewline
\isadelimproof
\ \ %
\endisadelimproof
\isatagproof
\isacommand{apply}\isamarkupfalse%
\ {\isacharparenleft}auto\ simp\ add{\isacharcolon}\ kleisli{\isacharunderscore}def\ option{\isacharunderscore}max{\isacharunderscore}assoc{\isacharparenright}\isanewline
\ \ \isacommand{apply}\isamarkupfalse%
\ {\isacharparenleft}simp\ add{\isacharcolon}\ option{\isacharunderscore}max{\isacharunderscore}commut{\isacharparenright}\isanewline
\ \ \isacommand{done}\isamarkupfalse%
\endisatagproof
{\isafoldproof}%
\isadelimproof
\isanewline
\endisadelimproof
\isanewline
\isacommand{lemma}\isamarkupfalse%
\ {\isacharparenleft}\isakeyword{in}\ delta{\isacharunderscore}gcounter{\isacharparenright}\ concurrent{\isacharunderscore}operations{\isacharunderscore}commute{\isacharcolon}\isanewline
\ \ \isakeyword{assumes}\ {\isachardoublequoteopen}xs\ prefix\ of\ i{\isachardoublequoteclose}\isanewline
\ \ \isakeyword{shows}\ {\isachardoublequoteopen}hb{\isachardot}concurrent{\isacharunderscore}ops{\isacharunderscore}commute\ {\isacharparenleft}node{\isacharunderscore}deliver{\isacharunderscore}messages\ xs{\isacharparenright}{\isachardoublequoteclose}\isanewline
\isadelimproof
\ \ %
\endisadelimproof
\isatagproof
\isacommand{using}\isamarkupfalse%
\ assms\isanewline
\ \ \isacommand{apply}\isamarkupfalse%
{\isacharparenleft}clarsimp\ simp{\isacharcolon}\ hb{\isachardot}concurrent{\isacharunderscore}ops{\isacharunderscore}commute{\isacharunderscore}def{\isacharparenright}\isanewline
\ \ \isacommand{apply}\isamarkupfalse%
{\isacharparenleft}unfold\ interp{\isacharunderscore}msg{\isacharunderscore}def{\isacharcomma}\ simp{\isacharparenright}\isanewline
\ \ \isacommand{done}\isamarkupfalse%
\endisatagproof
{\isafoldproof}%
\isadelimproof
\isanewline
\endisadelimproof
\isanewline
\isacommand{corollary}\isamarkupfalse%
\ {\isacharparenleft}\isakeyword{in}\ delta{\isacharunderscore}gcounter{\isacharparenright}\ counter{\isacharunderscore}convergence{\isacharcolon}\isanewline
\ \ \isakeyword{assumes}\ {\isachardoublequoteopen}set\ {\isacharparenleft}node{\isacharunderscore}deliver{\isacharunderscore}messages\ xs{\isacharparenright}\ {\isacharequal}\ set\ {\isacharparenleft}node{\isacharunderscore}deliver{\isacharunderscore}messages\ ys{\isacharparenright}{\isachardoublequoteclose}\isanewline
\ \ \ \ \ \ \isakeyword{and}\ {\isachardoublequoteopen}xs\ prefix\ of\ i{\isachardoublequoteclose}\isanewline
\ \ \ \ \ \ \isakeyword{and}\ {\isachardoublequoteopen}ys\ prefix\ of\ j{\isachardoublequoteclose}\isanewline
\ \ \ \ \isakeyword{shows}\ {\isachardoublequoteopen}apply{\isacharunderscore}operations\ xs\ {\isacharequal}\ apply{\isacharunderscore}operations\ ys{\isachardoublequoteclose}\isanewline
\isadelimproof
\endisadelimproof
\isatagproof
\isacommand{using}\isamarkupfalse%
\ assms\ \isacommand{by}\isamarkupfalse%
{\isacharparenleft}auto\ simp\ add{\isacharcolon}\ apply{\isacharunderscore}operations{\isacharunderscore}def\ intro{\isacharcolon}\ hb{\isachardot}convergence{\isacharunderscore}ext\ concurrent{\isacharunderscore}operations{\isacharunderscore}commute\isanewline
\ \ \ \ \ \ \ \ \ \ \ \ \ \ \ \ node{\isacharunderscore}deliver{\isacharunderscore}messages{\isacharunderscore}distinct\ hb{\isacharunderscore}consistent{\isacharunderscore}prefix{\isacharparenright}%
\endisatagproof
{\isafoldproof}%
\isadelimproof
\isanewline
\endisadelimproof
\isanewline
\isacommand{context}\isamarkupfalse%
\ delta{\isacharunderscore}gcounter\ \isakeyword{begin}\isanewline
\isanewline
\isacommand{sublocale}\isamarkupfalse%
\ sec{\isacharcolon}\ strong{\isacharunderscore}eventual{\isacharunderscore}consistency\ weak{\isacharunderscore}hb\ hb\ interp{\isacharunderscore}msg\isanewline
\ \ {\isachardoublequoteopen}{\isasymlambda}ops{\isachardot}\ {\isasymexists}xs\ i{\isachardot}\ xs\ prefix\ of\ i\ {\isasymand}\ node{\isacharunderscore}deliver{\isacharunderscore}messages\ xs\ {\isacharequal}\ ops{\isachardoublequoteclose}\ {\isachardoublequoteopen}{\isasymlambda}\ x{\isachardot}\ None{\isachardoublequoteclose}\isanewline
\isadelimproof
\ \ %
\endisadelimproof
\isatagproof
\isacommand{apply}\isamarkupfalse%
{\isacharparenleft}standard{\isacharsemicolon}\ clarsimp{\isacharparenright}\isanewline
\ \ \ \ \ \ \isacommand{apply}\isamarkupfalse%
{\isacharparenleft}auto\ simp\ add{\isacharcolon}\ hb{\isacharunderscore}consistent{\isacharunderscore}prefix\ drop{\isacharunderscore}last{\isacharunderscore}message\ node{\isacharunderscore}deliver{\isacharunderscore}messages{\isacharunderscore}distinct\ concurrent{\isacharunderscore}operations{\isacharunderscore}commute{\isacharparenright}\isanewline
\ \ \ \isacommand{apply}\isamarkupfalse%
{\isacharparenleft}metis\ {\isacharparenleft}full{\isacharunderscore}types{\isacharparenright}\ interp{\isacharunderscore}msg{\isacharunderscore}def\ delta{\isacharunderscore}gcounter{\isacharunderscore}op{\isachardot}elims{\isacharparenright}\isanewline
\ \ \isacommand{using}\isamarkupfalse%
\ drop{\isacharunderscore}last{\isacharunderscore}message\ \isacommand{apply}\isamarkupfalse%
\ blast\isanewline
\ \ \isacommand{done}\isamarkupfalse%
\endisatagproof
{\isafoldproof}%
\isadelimproof
\isanewline
\endisadelimproof
\isacommand{end}\isamarkupfalse%
\isanewline
\isadelimtheory
\isanewline
\endisadelimtheory
\isatagtheory
\isacommand{end}\isamarkupfalse%
\endisatagtheory
{\isafoldtheory}%
\isadelimtheory
\endisadelimtheory
\end{isabellebody}%

\section{Restricted $\delta$-state G-Set \CRDT}
\begin{isabellebody}
\isacommand{locale}\isamarkupfalse%
\ delta{\isacharunderscore}gset\ {\isacharequal}\ network{\isacharunderscore}with{\isacharunderscore}ops\ {\isacharunderscore}\ delta{\isacharunderscore}gset{\isacharunderscore}op\ {\isachardoublequoteopen}{\isacharbraceleft}{\isacharbraceright}{\isachardoublequoteclose}\isanewline
\isanewline
\isacommand{lemma}\isamarkupfalse%
\ {\isacharparenleft}\isakeyword{in}\ delta{\isacharunderscore}gset{\isacharparenright}\ {\isacharbrackleft}simp{\isacharbrackright}\ {\isacharcolon}\ {\isachardoublequoteopen}delta{\isacharunderscore}gset{\isacharunderscore}op\ x\ {\isasymrhd}\ delta{\isacharunderscore}gset{\isacharunderscore}op\ y\ {\isacharequal}\ delta{\isacharunderscore}gset{\isacharunderscore}op\ y\ {\isasymrhd}\ delta{\isacharunderscore}gset{\isacharunderscore}op\ x{\isachardoublequoteclose}\isanewline
\isadelimproof
\ \ %
\endisadelimproof
\isatagproof
\isacommand{apply}\isamarkupfalse%
\ {\isacharparenleft}auto\ simp\ add{\isacharcolon}\ kleisli{\isacharunderscore}def{\isacharparenright}\isanewline
\isacommand{done}\isamarkupfalse%
\endisatagproof
{\isafoldproof}%
\isadelimproof
\isanewline
\endisadelimproof
\isanewline
\isacommand{lemma}\isamarkupfalse%
\ {\isacharparenleft}\isakeyword{in}\ delta{\isacharunderscore}gset{\isacharparenright}\ concurrent{\isacharunderscore}operations{\isacharunderscore}commute{\isacharcolon}\isanewline
\ \ \isakeyword{assumes}\ {\isachardoublequoteopen}xs\ prefix\ of\ i{\isachardoublequoteclose}\isanewline
\ \ \isakeyword{shows}\ {\isachardoublequoteopen}hb{\isachardot}concurrent{\isacharunderscore}ops{\isacharunderscore}commute\ {\isacharparenleft}node{\isacharunderscore}deliver{\isacharunderscore}messages\ xs{\isacharparenright}{\isachardoublequoteclose}\isanewline
\isadelimproof
\ \ %
\endisadelimproof
\isatagproof
\isacommand{using}\isamarkupfalse%
\ assms\isanewline
\ \ \isacommand{apply}\isamarkupfalse%
{\isacharparenleft}clarsimp\ simp{\isacharcolon}\ hb{\isachardot}concurrent{\isacharunderscore}ops{\isacharunderscore}commute{\isacharunderscore}def{\isacharparenright}\isanewline
\ \ \isacommand{apply}\isamarkupfalse%
{\isacharparenleft}unfold\ interp{\isacharunderscore}msg{\isacharunderscore}def{\isacharcomma}\ simp{\isacharparenright}\isanewline
\isacommand{done}\isamarkupfalse%
\endisatagproof
{\isafoldproof}%
\isadelimproof
\isanewline
\endisadelimproof
\isanewline
\isacommand{corollary}\isamarkupfalse%
\ {\isacharparenleft}\isakeyword{in}\ delta{\isacharunderscore}gset{\isacharparenright}\ set{\isacharunderscore}convergence{\isacharcolon}\isanewline
\ \ \isakeyword{assumes}\ {\isachardoublequoteopen}set\ {\isacharparenleft}node{\isacharunderscore}deliver{\isacharunderscore}messages\ xs{\isacharparenright}\ {\isacharequal}\ set\ {\isacharparenleft}node{\isacharunderscore}deliver{\isacharunderscore}messages\ ys{\isacharparenright}{\isachardoublequoteclose}\isanewline
\ \ \ \ \ \ \isakeyword{and}\ {\isachardoublequoteopen}xs\ prefix\ of\ i{\isachardoublequoteclose}\isanewline
\ \ \ \ \ \ \isakeyword{and}\ {\isachardoublequoteopen}ys\ prefix\ of\ j{\isachardoublequoteclose}\isanewline
\ \ \ \ \isakeyword{shows}\ {\isachardoublequoteopen}apply{\isacharunderscore}operations\ xs\ {\isacharequal}\ apply{\isacharunderscore}operations\ ys{\isachardoublequoteclose}\isanewline
\isadelimproof
\endisadelimproof
\isatagproof
\isacommand{using}\isamarkupfalse%
\ assms\ \isacommand{by}\isamarkupfalse%
{\isacharparenleft}auto\ simp\ add{\isacharcolon}\ apply{\isacharunderscore}operations{\isacharunderscore}def\ intro{\isacharcolon}\ hb{\isachardot}convergence{\isacharunderscore}ext\ concurrent{\isacharunderscore}operations{\isacharunderscore}commute\isanewline
\ \ \ \ \ \ \ \ \ \ \ \ \ \ \ \ node{\isacharunderscore}deliver{\isacharunderscore}messages{\isacharunderscore}distinct\ hb{\isacharunderscore}consistent{\isacharunderscore}prefix{\isacharparenright}%
\endisatagproof
{\isafoldproof}%
\isadelimproof
\isanewline
\endisadelimproof
\isanewline
\isacommand{context}\isamarkupfalse%
\ delta{\isacharunderscore}gset\ \isakeyword{begin}\isanewline
\isanewline
\isacommand{sublocale}\isamarkupfalse%
\ sec{\isacharcolon}\ strong{\isacharunderscore}eventual{\isacharunderscore}consistency\ weak{\isacharunderscore}hb\ hb\ interp{\isacharunderscore}msg\isanewline
\ \ {\isachardoublequoteopen}{\isasymlambda}ops{\isachardot}\ {\isasymexists}xs\ i{\isachardot}\ xs\ prefix\ of\ i\ {\isasymand}\ node{\isacharunderscore}deliver{\isacharunderscore}messages\ xs\ {\isacharequal}\ ops{\isachardoublequoteclose}\ {\isachardoublequoteopen}{\isacharbraceleft}{\isacharbraceright}{\isachardoublequoteclose}\isanewline
\isadelimproof
\ \ %
\endisadelimproof
\isatagproof
\isacommand{apply}\isamarkupfalse%
{\isacharparenleft}standard{\isacharsemicolon}\ clarsimp{\isacharparenright}\isanewline
\ \ \ \ \ \ \isacommand{apply}\isamarkupfalse%
{\isacharparenleft}auto\ simp\ add{\isacharcolon}\ hb{\isacharunderscore}consistent{\isacharunderscore}prefix\ drop{\isacharunderscore}last{\isacharunderscore}message\ node{\isacharunderscore}deliver{\isacharunderscore}messages{\isacharunderscore}distinct\ concurrent{\isacharunderscore}operations{\isacharunderscore}commute{\isacharparenright}\isanewline
\ \ \ \isacommand{apply}\isamarkupfalse%
{\isacharparenleft}metis\ {\isacharparenleft}full{\isacharunderscore}types{\isacharparenright}\ interp{\isacharunderscore}msg{\isacharunderscore}def\ delta{\isacharunderscore}gset{\isacharunderscore}op{\isachardot}elims{\isacharparenright}\isanewline
\ \ \isacommand{using}\isamarkupfalse%
\ drop{\isacharunderscore}last{\isacharunderscore}message\ \isacommand{apply}\isamarkupfalse%
\ blast\isanewline
\ \ \isacommand{done}\isamarkupfalse%
\endisatagproof
{\isafoldproof}%
\isadelimproof
\isanewline
\endisadelimproof
\isacommand{end}\isamarkupfalse%
\isanewline
\isadelimtheory
\isanewline
\endisadelimtheory
\isatagtheory
\isacommand{end}\isamarkupfalse%
\endisatagtheory
{\isafoldtheory}%
\isadelimtheory
\endisadelimtheory
\end{isabellebody}%

  \newpage
  \addcontentsline{toc}{chapter}{Bibliography}
  \bibliography{thesis.bib}
\end{document}